%% file: Constraint_Tensor_Null_Case.tex
\documentclass[11pt]{article}

\usepackage[absolute,overlay]{textpos}

\usepackage{float}
\usepackage[dvips]{graphicx}
\usepackage{psfrag}
\usepackage{amsmath,amsthm,amsfonts,amssymb,amscd}
\usepackage[table]{xcolor}
\usepackage{mathrsfs}
\usepackage{upgreek}
\usepackage{stackrel}
\usepackage{tipa}
\usepackage{accents}
\usepackage[multiple]{footmisc}

\usepackage{multirow} %for tables
\usepackage{booktabs}

\numberwithin{equation}{section}

\usepackage{multicol}

\definecolor{SmartBlue}{RGB}{51, 51, 255}

\usepackage{hyperref}

\hypersetup{filecolor=magenta}
\hypersetup{colorlinks=true}
\hypersetup{urlcolor=webbrown}
\hypersetup{linkcolor=SmartBlue}
\hypersetup{pdfpagemode=UseNone}
\hypersetup{citecolor=SmartBlue}
\input{paper-config-newcommands2}

\usepackage[a4paper,top=25mm,bottom=25mm,left=23mm, right=23mm,bindingoffset=0mm]{geometry} %This package resizes the margins to meet La Trobe's guidelines

\title{The constraint tensor for null hypersurfaces}
\author{Miguel Manzano\footnote{miguelmanzano06@usal.es}\hspace{0.17cm} and Marc Mars\footnote{marc@usal.es}\\ \\
Instituto de F\'{\i}sica Fundamental y Matem\'aticas, IUFFyM\\
Universidad de Salamanca
}
\usepackage{bbm}

\begin{document}

% Adjust space above and below equations
\setlength{\abovedisplayskip}{4pt} % Adjust as needed
\setlength{\belowdisplayskip}{4pt} % Adjust as needed

\maketitle

\vspace{-0.8cm}

\begin{abstract}
In this work we provide a definition of the constraint tensor of a null hypersurface data which is completely explicit in the extrinsic geometry of the hypersurface.\ The definition is fully covariant and applies for any topology of the hypersurface.\ For data embedded in a spacetime, the constraint tensor coincides with the pull-back of the ambient Ricci tensor.\ As applications of the results, we find three geometric quantities on any transverse submanifold $S$ of the data with remarkably simple gauge behaviour, and prove that the restriction of the constraint tensor to $S$ takes a very simple form in terms of them.\ We also obtain an identity that generalizes the standard near horizon equation of isolated horizons to totally geodesic null hypersurfaces with any topology.\ Finally, we prove that when a null hypersurface has product topology, its extrinsic curvature can be uniquely reconstructed from the constraint tensor plus suitable initial data on a cross-section.
\end{abstract}

%\pagebreak
%
%\tableofcontents
%
%\pagebreak

\section{Introduction}\label{sec:Intro:CT}

The curvature of a spacetime and the geometry of an embedded hypersurface are well-known to be linked by a set of \textit{constraint equations}, which in the case of 
%For 
non-degenerate hypersurfaces are 
%, these equations are 
the standard Gauss and Codazzi equations.\  
%\cite{lefloch2007definition, mars1993geometry},
In the null case, related results have been obtained in different contexts and with different degrees of generality.\ In particular, a formalism that puts the null geometry to the forefront is the Newman-Penrose formalism \cite{ahsan2019potential, newman1962approach}.\  Another classic approach is the use of privileged coordinates, such as the Bondi-Sachs coordinates \cite{bondi1962gravitational, sachs1962gravitational}, or adapted Gaussian null coordinates
%, first introduced in 
\cite{moncrief1983symmetries}.\ Many other works have explored the connection between the geometry of null hypersurfaces and the ambient curvature.\ A representative, but necessarily far from exhaustive, list of references is 
\cite{ashtekar2002geometry, ashtekar2000isolated, bartnik1997einstein,  christodoulou2008formationblackholesgeneral,czimek2022obstruction,  galloway2000maximum, galloway2004null, gourgoulhon20073+, gourgoulhon20123+,  gourgoulhon20063+, jaramillo2009isolated, jezierski2004geometry, jezierski2000geometry, robson1973null, schouten1954ricci, wald1984general, winicour2013affine}.\ An example of a constraint equation 
%of this type 
is the well-known null Raychaudhuri equation (see e.g.\ \cite{carter2012gravitation, gourgoulhon20063+,wald1984general}), which relates the 
%components of a 
spacetime Ricci tensor along the null generators with the expansion scalar and the shear tensor of an embedded null hypersurface.\

One of the main motivations for analysing the interplay between spacetime curvature and hypersurface geometry is the study of the Einstein equations.\ This is so because the combination of the Einstein and the constraint equations  
%the connection between spacetime curvature and hypersurface geometry 
enables a decomposition of the former %Einstein equations 
into a set of equalities \textit{on} the hypersurface plus evolution equations along a direction transverse to the hypersurface.\ This splitting of the Einstein equations is actually the key milestone on which the initial value problems rely.\ 
In particular, in the Cauchy problem \cite{choquet1952theoreme, ringstrom2009cauchy} 
the data are, depending on the approach, either
the full ambient metric and its first transverse derivative 
%are prescribed 
on a spacelike hypersurface,  
or its first and second fundamental forms, in both cases subject to 
%together with 
certain constraint equations.\ In the characteristic initial value problem the initial data is given on two intersecting null hypersurfaces \cite{cabet2014characteristic, Caciotta2005global, choquet2011cauchy, chrusciel2012many, hilditch2020revisiting, klainerman2002evolution, luk2012local, mars2023covariant, Mars2023first, rendall1990reduction,  sachs1962characteristic}.\ This null data is also subject to  constraints	which can vary in form and in number depending on which specific data one provides (see the discussion in \cite{mars2023covariant}).\ 
In any case, the 
%data is forced to satisfy suitable 
constraint equations guarantee that one can construct a spacetime fulfilling the vacuum Einstein equations, and where the hypersurfaces happen to be embedded.\ 

Specifically, the constraint equations are identities that relate some components of spacetime curvature tensors with tensor fields encoding the intrinsic/extrinsic geometry of a hypersurface.\ In the case of a null hypersurface $\nullhyp$ embedded in a spacetime $(\M,g)$, they have usually been obtained under suitable topological assumptions  \cite{ashtekar2002geometry, ashtekar2000isolated, gourgoulhon20073+, gourgoulhon20123+, gourgoulhon20063+, katona2024uniqueness, kunduri2013classification, li2016transverse, mars2018multiple,mars2018nearhorizon}, namely that the null hypersurface has a product topology $S \times \mathbb{R}$ where $S$ is a spacelike cross-section, and sometimes under additional restrictions on the spacetime,
%
%.\
e.g.\ by making use of a $\{3+1\}$ slicing of the spacetime\footnote{I.e.\ a foliation by spacelike or null hypersurfaces.}  which enables the extension of tensor fields away from the original hypersurface.\ By exploiting the additional structure arising from $\nullhyp$ being ruled by a collection of lightlike geodesics, some of these identities have been written as 
% allows one to write these identities 
%in the form of 
evolution equations for geometric properties of the cross-sections of $\nullhyp$ (e.g.\ the torsion one-form and the extrinsic curvature) along the null generators.\ 
This approach, while 
%this approach is 
advantageous in many circumstances, cannot provide fully general results because 
%, since 
%generically 
the topological assumptions above can only be granted on local domains of a given spacetime.\  
Furthermore, the identities 
%, in addition, 
are \textit{not} covariant in the sense that the construction depends on the coordinates or on a choice of foliation.\
%the foliations of the spacetime and the null hypersurface.\ 
This makes it necessary to analyse how the equations and the geometric objects transform under changes of  the coordinates and/or the foliation.\

The purpose of the present work is to demonstrate that much of the structure 
induced on a null hypersurface under the assumptions described above 
%in the previously mentioned works 
can be preserved 
%for any dimension, 
for any topology of the hypersurface.\   
In fact, our results hold for 
%so-called 
\textit{abstract} (or detached) null hypersurfaces, i.e.\ null hypersurfaces that are \textit{not} yet considered as embedded in any ambient space.\   
The \textit{fully covariant} approach 
%identities 
%(despite dealing with null hypersurfaces) 
that we follow allows us to obtain 
%lead us to 
\textit{completely general identities} relating the geometry of an abstract null hypersurface with 
the curvature of a spacetime (encoded in a tensor field on the hypersurface).\ Such approach, in addition, is \textit{purely geometric}, which makes it easily adaptable to any specific situation at hand 
%one wants to address 
(e.g.\ to scenarios where the hypersurface has a product topology, or to the case when the extrinsic curvature is defined w.r.t.\ a specific transverse vector).\  
Our conclusions have several potential applications, some of which are detailed below.\ 
%which we will detailed next.\ 
To facilitate this discussion, we will first review some details about abstract hypersurfaces.

The framework that enables the description of 
%the geometry of 
hypersurfaces independently of any ambient space is the so-called \textit{formalism of hypersurface data} \cite{mars2013constraint, mars2020hypersurface} 
%(see also \cite{ manzano2023PhD, manzano2021null,  manzano2022general, manzano2023matching, manzano2023field, mars2013constraint, mars2020hypersurface, mars2024abstract, mars2023covariant, Mars2023first, mars2024transverseI, mars2024transverseII, nolan2019first}).\ 
(see also \cite{manzano2023matching, manzano2023field, mars2024abstract, mars2023covariant, Mars2023first, mars2024transverseI, mars2024transverseII}).\ 
In this formalism,
the geometry of a hypersurface is codified 
%\textit{abstractly} (i.e.\ 
in a detached way from any ambient space
%)  
in terms of a data set $\hypdata$, where $\N$ is a smooth manifold, $\{\gamma,\bY\}$ are symmetric $(0,2)$-tensor fields, $\ellc$ is a covector field, and $\elltwo$ is a scalar field satisfying a certain restriction (later we give the precise definition).\ When the data is embedded in a spacetime $(\M,g)$ with embedding $\phi$ and a choice of rigging\footnote{Given an embedded hypersurface $\mathcal{H}$, a rigging vector is a vector field along $\mathcal{H}$, everywhere transverse to it \cite{schouten1954ricci}.} $\rig$, the full metric $g$ along $\phi(\N)$ can be reconstructed from $\metdataa$, and $\{\ellc,\elltwo\}$ coincide with the pull-back of the one-form $g(\rig,\cdot)$ and the norm of the rigging respectively.\ The tensor $\bY$, on the other hand, encodes the pull-back to $\N$ of first transverse derivatives of $g$.\ Thus, $\metdataa$ and $\bY$ codify the intrinsic and extrinsic geometry of the hypersurface respectively.\
Although the hypersurface data formalism applies for hypersurfaces of any causal character, in this paper we shall only consider the null case.\ 

The hypersurface data formalism constitutes a very convenient framework to address initial value problems, as it does not require a priori the existence of the ambient space.\ In the works \cite{mars2023covariant,Mars2023first} on the 
%context of the 
characteristic initial value problem, 
%and in the spirit of encoding information of the spacetime curvature at a purely abstract level, 
%in   
%the authors introduced 
the so-called \textit{constraint tensor} $\Rtensor$ on an abstract null hypersurface $\N$ was introduced.\ This is a symmetric $(0,2)$-tensor with the property that, when 
%for null data 
%$\hypdata$ 
$\N$ is embedded in a spacetime $(\M,g)$, it corresponds to the pull-back to $\N$ of the  Ricci tensor of $g$.\ 
This definition was already fully covariant.\ However, by the way how this tensor was constructed the dependence on the extrinsic curvature tensor $\bY$ is not explicit (the connection and curvature tensors used in \cite{mars2023covariant,Mars2023first} to write down the covariant expressions depend on $\bY$).\ While this was sufficient for the purposes in those papers, this choice is not convenient in many other contexts.\

Our main aim in this paper is to 
%In this paper, we 
find an alternative form for $\Rtensor$ which shows the full dependence on $\bY$.\ 
We achieve this in two steps.\ First, we consider null embedded hypersurface data and compute the \textit{transverse-tangent-tangent-tangent} 
%$1$-transverse, $3$-tangent 
and the \textit{fully tangent} components of the ambient Riemann tensor in terms of the so-called ``metric hypersurface connection"  $\nablao$ \cite{mars2013constraint,mars2020hypersurface}.\ This connection is independent of the tensor $\bY$.\ %However, this is not quite sufficient for our purpose, namely displaying the dependence on $\bY$ in the most convenient way.\ In particular, 
Then, we work out a number of identities (see Section \ref{secAbstractRicci}), which eventually allow us to write   
the constraint tensor $\Rtensor$, and its contraction with a null generator $n$ of $\N$, in terms of Lie derivatives of $\bY$ and $\bs{r }\defi \bY(n,\cdot)$ along $n$.\ The only quantity that \textit{does not} satisfy a transport equation is  the scalar $\Q \defi - \bs{r}(n)$.\ Instead, $\Q$ satisfies an equation that generalizes the null Raychaudhuri equation to the detached setting.\
We emphasize that these results hold for
%to 
\textit{completely general topologies of the null hypersurface}, and that the expressions are fully covariant, i.e.\ they do not depend on any specific choice of coordinates and/or of a foliation on $\N$ (whenever it exists).\

We have already mentioned that this form of the constraint tensor has several potential applications.\ In this paper we explore three of them.\ In Section \ref{secRiemoAB:gauge:inv} we prove that the pull-back of $\Rtensor$ to any non-degenerate codimension-one submanifold $S$ of $\N$ (not necessarily a cross-section) can be written in terms of three geometric quantities, namely a one-form $\bomegap$ and two symmetric $(0,2)$-tensors $\bsecp$, $\bthip$.\ The remarkable property of the tensors  $\{\bomegap,\bsecp,\bthip\}$ is that they exhibit very simple behaviour under changes in the choice of the rigging.\ The tensors $\bomegap$ and $\bsecp$ can be interpreted as generalizations of the standard torsion one-form and second fundamental form of $S$, defined  w.r.t.\ a null and normal rigging.\ The quantity $\bthip$ codifies curvature information of the hypersurface (namely derivatives of the extrinsic curvature), and it plays a key role in the study of horizons \cite{manzano2023PhD}.\ 
%its geometric meaning is less clear.\ 
%and 
%its properties deserve further investigation.\ 
As a second application of the results in this paper, %in Section 
in Section \ref{sec:GME} we obtain a generalization of the well-known near horizon equation of isolated horizons \cite{ashtekar2000generic, ashtekar2002geometry, ashtekar2000isolated, gourgoulhon20063+, jaramillo2009isolated, krishnan2002isolated}\footnote{See also the works \cite{mars2018multiple,mars2018nearhorizon, mars2019multiple} on so-called multiple Killing horizons.}, which is valid for \textit{any} totally geodesic null hypersurface equipped with a privileged null tangent vector (irrespectively of the topology of the hypersurface).\  
%and of the spacetime where it is  embedded).\ 
This equation holds on any codimension-one non-degenerate submanifold of the hypersurface (again, not necessarily a cross-section), and can be written in terms of $\{\bomegap,\bsecp,\bthip\}$.\ In fact, in an upcoming work \cite{manzano2023master} we will combine the results of the present paper with those in \cite{manzano2023field} to derive an equation
%, which we call the generalized master equation, 
that generalizes the near horizon equation of isolated horizons to \textit{any null hypersurface admitting a privileged null tangent vector}, allowing also for the presence of zeroes.\

%Another application of our results can be found 
Our third application is presented 
in Section \ref{sec:application}.\ Here we restrict ourselves to product topology 
%In Proposition  \ref{proposition:integrate:Y:tensor} therein, we consider an abstract null hypersurface $\N$ with product topology 
and prove 
%(Proposition  \ref{proposition:integrate:Y:tensor})
that, given a constraint tensor $\Rtensor$ plus initial data on a cross-section, one can uniquely reconstruct the tensor $\bY$ everywhere, provided %(and only then) 
that the detached version of the Raychaudhuri equation admits a solution.\ 
Although the result is stated for cross-sections, it is obviously true in sufficient local domains of an arbitrary $\N$.\ The result, presented in Proposition \ref{proposition:integrate:Y:tensor},  is  
%Although this result only holds under the hypothesis that $\N$ has product topology, it still holds locally near any Riemannian submanifold of $\N$.\ Proposition \ref{proposition:integrate:Y:tensor} is extremely 
useful in any situation where the extrinsic curvature plays a role, e.g.\ 
%for instance, 
in the contexts of initial value problems (see e.g.\ \cite{mars2023covariant,Mars2023first}) or spacetime matching (see e.g.\ \cite{manzano2021null,manzano2022general,manzano2023matching}).\ In the former, it allows one to reduce the initial data to metric hypersurface data plus suitable information on a cross-section.\ In the spacetime matching context, 
%on the other hand, since the matter content of thin shells is precisely given by the jump in extrinsic curvature, 
Proposition \ref{proposition:integrate:Y:tensor} determines the matter and impulsive gravitational wave content of any null shell with product topology from information of the curvature of the spacetimes to be matched plus initial conditions on a cross-section of the shell.\ %This is very useful, as it allows characterizing null shells through the jump in extrinsic curvature on a single section.
It is also worth mentioning that the results in this paper have already been used successfully in the recent works \cite{mars2024transverseI,mars2024transverseII}, where  the full transverse expansion of a spacetime metric on a general null hypersurface is studied.

The organization of the paper is as follows.\ Section \ref{sec:FHD:Prelim} is devoted to introducing basic concepts and results of the formalism of hypersurface data.\ In Section \ref{seclevicivitaonS} we analyse some properties of codimension-one non-degenerate submanifolds embedded in null hypersurface data.\ In Section \ref{secAbstractRicci} we provide a definition for the constraint tensor $\Rtensor$ where all dependence on $\bY$ is explicit.\ We then find the contractions of $\Rtensor$ with a null generator (Section \ref{sec:Constraint(n,-)}) and provide its pull-back 
$\Rtensor_{\parallel}$ to a Riemannian codimension-one submanifold $(S,h)$ in terms of the Ricci tensor of $h$ (Section \ref{secRiemoAB}).\ The paper concludes with Sections \ref{secRiemoAB:gauge:inv}, \ref{sec:GME} and \ref{sec:application}.\ As already mentioned, in the first one we introduce the quantities $\{\bomegap,\bsecp,\bthip\}$ and write the tensor $\Rtensor_{\parallel}$ in terms of them.\ In the second one, we derive a generalized near horizon equation for totally geodesic null hypersurfaces.\ Finally, in Section \ref{sec:application} we prove Proposition \ref{proposition:integrate:Y:tensor}, described above.\ We also include three appendices.\ In Appendix \ref{secGauss} we derive a generalized form of a Gauss-type identity, valid for a very general class of embedded hypersurfaces and ambient manifolds.\
%provided that both of them are equipped with a torsion-free connection.\ 
This is used in the main text to compute the pull-back of the constraint tensor to a non-degenerate submanifold (Sections \ref{seclevicivitaonS} and  \ref{secAbstractRicci}).\ 
%Appendix \ref{app:curvature:null} is devoted to various identities involving the curvature and Ricci tensors of a torsion-free connection constructed from a metric data set.\ Lastly, 
In the paper,  
%main text, 
we also need to use several properties of the curvature tensor of the connection $\nablao$. Appendix \ref{app:curvature:null} is devoted to deriving them.\ Finally, in Appendix \ref{sec:gauge-fix:res} we examine the gauge behaviour of several geometric quantities on $\N$. 

\subsection{Notation and conventions}\label{sec:notation}

In this paper, all manifolds are smooth, connected and without boundary. Given a manifold $\M$ we use $\mathcal{F}\lp\mathcal{M}\rp\defi C^{\infty}\lp\mathcal{M},\mathbb{R}\rp$ and $\mathcal{F}^{\star}\lp\mathcal{M}\rp\subset\mathcal{F}\lp\mathcal{M}\rp$ for the subset of no-where zero functions.  The tangent bundle is denoted 
by $T\mathcal{M}$ and $\Gamma\lp T\mathcal{M}\rp$ is the set
of sections (i.e.\ vector fields).  We use $\pounds$, $d$ for the Lie derivative and exterior derivative. Both tensorial and abstract index notation will be employed. 
%used depending on convenience. 
We work in arbitrary dimension $\mathfrak{n}$ and use the following sets of indices:
\begin{equation}
\label{notation}
\alpha,\beta,...=0,1,2,...,\mathfrak{n};\qquad a,b,...=1,2,...,\mathfrak{n};\qquad A,B,...=2,...,\mathfrak{n}.
\end{equation}
When index-free notation is used (and only then) we shall distinguish
covariant tensors with boldface.\ As usual, parenthesis (resp.\ brackets) denote symmetrization (resp.\ antisymmetrization) of indices.\ The symmetrized tensor product is defined by $A\otimes_s B\equiv\frac{1}{2}(A\otimes B+B\otimes A)$.\ We write $\text{tr}_BA$ for the trace of a symmetric $(0,2)$-tensor $A$ w.r.t.\
a $(2,0)$-tensor $B$.\ In any Lorentzian manifold $(\mathcal{M},g)$,
%the scalar product of two vectors is written both as $g(X,Y)$ or $\la X,Y\rag$, and 
we use $g^{\sharp}$, $\nabla$  for the inverse metric and Levi-Civita derivative of $g$ respectively.\ We also denote the Riemann and Ricci tensors of $g$ by $R^{\mu}{}_{\alpha\beta\nu}$,  $R_{\alpha\beta}$ (or  $\textbf{Riem}$,  $\textbf{Ric}$ in index-free notation).\ 
Our signature convention for Lorentzian manifolds $\lp \mathcal{M},g\rp$ is $(-,+, ... ,+)$.

\section{Formalism of hypersurface data}\label{sec:FHD:Prelim}
In this section we introduce all necessary aspects of the formalism of hypersurface data.\   %exploited throughout the paper, namely the \textit{formalism of hypersurface data}. 
This work focuses on null hypersurfaces, so we only present the formalism in the null case.\ 
New results are all demonstrated, while for the known ones we give a reference.\ %do not provide proof, while the new ones are suitably demonstrated.
For further details on the 
%hypersurface data 
formalism we refer to \cite{manzano2023matching, manzano2023field,mars2013constraint, mars2020hypersurface,  mars2024abstract, mars2023covariant, Mars2023first, mars2024transverseI, mars2024transverseII}.

%\subsection{Metric hypersurface data}\label{subsec:MHD}

%The hypersurface data formalism is a framework to codify, at an abstract level, the information concerning the  geometry of a hypersurface $\N$ of a Lorentzian manifold. By ``abstract'' we mean that one can work without making reference to any embedding of $\N$ into an ambient space. The basic notion if the following AÑADIR CITAS
%\begin{definition}[Metric hypersurface data]
%	\label{defMHD}
A $4$-tuple  $\metdata$ consisting of
a smooth $\n$-manifold $\N$ endowed with a symmetric $(0,2)$-tensor field $ \gamma $ of signature $(0,+,...,+)$, a covector field $\ellc$ and a scalar function $\ell^{(2)}$ defines 
\textit{null metric hypersurface data} provided that the square $(\n+1)$-matrix
\begin{equation}
\bs{\A}\defi \lp \hspace{-0.1cm}
\begin{array}{cc}
\gamma_{ab} & \ell_a\\
\ell_b & \elltwo
\end{array}
\hspace{-0.1cm}\rp
\end{equation}
is non-degenerate everywhere on $\N$.\ 
%
%  
%the symmetric $(0,2)$-tensor 
%$\bs{\mathcal{A}}\vert_p$ on $T_p\mathcal{N}\times\mathbb{R}$ given by
%$$\ld\bs{\mathcal{A}}\rv_p\lp\lp W,a\rp,\lp Z,b\rp\rp\defi \ld \gamma \rv_p\lp W,Z\rp+a\ellc\vert_p\lp Z\rp+ b\ellc\vert_p\lp W\rp+a b \ell^{(2)}\vert_p,\quad W,Z\in T_p\mathcal{N},\quad a,b \in\mathbb{R}$$  
%is non-degenerate at every $p\in\mathcal{N}$.\ 
%
%
%\end{definition}
\textit{Null hypersurface data} $\hypdata$ is null metric hypersurface data equipped with an additional symmetric $(0,2)$-tensor $\bY$.\ From given null metric hypersurface data one can uniquely define a symmetric $(2,0)$-tensor field $P$ and a vector field $n$ as the entries of the inverse of $\bs{\A}$ \cite{mars2013constraint}, namely
\begin{equation}
\label{ambientinversemetric}	\bs{\A}^{-1}\defi \lp \hspace{-0.1cm}
	\begin{array}{cc}
		P^{ab} & n^a\\
		n^b & 0
	\end{array}
	\hspace{-0.1cm}\rp.
\end{equation} 
By construction, $P$ and $n$ %definition
%they 
satisfy

\vspace{-0.55cm}

\noindent
\begin{minipage}[t]{0.175\textwidth}
	\begin{align}
		\gamma_{ab} n^b & = 0, \label{prod1} 
	\end{align}
\end{minipage}
\hfill
\hfill
\begin{minipage}[t]{0.17\textwidth}
	\begin{align}
		\ell_a n^a & = 1, \label{prod2}  %\\
	\end{align}
\end{minipage}
\hfill
\begin{minipage}[t]{0.3\textwidth}
	\begin{align}
		P^{ab} \ell_b + \elltwo n^a & = 0,  \label{prod3} %\\
	\end{align}
\end{minipage}
\hfill
\begin{minipage}[t]{0.3\textwidth}
	\begin{align}
		P^{ab} \gamma_{bc} + n^a \ell_c & = \delta^a_c. \label{prod4}
	\end{align}
\end{minipage}

\vspace{0.1cm}

The vector $n$ is no-where zero by \eqref{prod2}, and the radical $\text{Rad}\gamma\vert_p\defi\{X\in T_p\N\hspace{0.05cm}\vert\hspace{0.05cm} \gamma(X,\cdot)=0\}$ of the  \textit{degenerate} tensor 
%By \eqref{prod1}-\eqref{prod2}, 
$\gamma$ is spanned by $n$ (cf.\ \eqref{prod1}).\ 
%However, its signature is partially restricted from the condition that $\bs{\A}$ is non-degenerate. Specifically, one can prove \cite{mars2020hypersurface} that the set of vectors anhihilated by $\gamma$ (i.e. the {\em radical} of $\gamma$, defined as
%$\textup{Rad}\gamma\vert_p \defi \{X\in T_p \N \spc\vert\spc\gamma(X,\cdot)=0\}$) is either zero- or one-dimensional at every point
%$p \in \N$. The latter case occurs if and only if $\ntwo\vert_p=0$ and then $n\vert_p$ is non-zero and defines the degenerate direction of $\gamma\vert_p$, i.e.
%$\textup{Rad}\gamma\vert_p=\la n\vert_p\ra$. A point here this happens is called null, and non-null otherwise.
We call the integral curves of $n$   \textit{generators} of $\N$.\ %, or \textit{generators} for short.\ 

It is useful to introduce the following tensor fields \cite{mars2020hypersurface}:

\vspace{-0.6cm}

%\begin{small}
\noindent
\begin{minipage}[t]{0.55\textwidth}
	\begin{align}
		\label{threetensors} \bF & \defi \frac{1}{2} d \ellc, & \bm{\sone} &\defi  \bF(n,\cdot), & \bU  &\defi  \frac{1}{2}\pounds_{n} \gamma ,
	\end{align}
\end{minipage}
%\begin{minipage}[t]{0.59\textwidth}
%\begin{align}
%	\label{threetensors} \bF & \defi \frac{1}{2} d \ellc, & \bm{\sone} &\defi  \bF(n,\cdot), & \bU  &\defi  \frac{1}{2}\pounds_{n} \gamma  + \ellc \otimes_s d \ntwo,
%\end{align}
%\end{minipage}%
\hfill
%\vrule
\hfill
\begin{minipage}[t]{0.44\textwidth}
	\begin{align}
		\label{defY(n,.)andQ}
		\bs{\Yn}&\defi\bY(n,\cdot),& \Q &\defi -\bY(n,n).
	\end{align}
\end{minipage}

\vspace{-0.cm}

%\begin{align}
%	\bF & \defi \frac{1}{2} d \ellc, \quad \quad \quad
%	\bm{\sone} \defi  \bF(n,\cdot), \quad \quad \quad
%	\bU  \defi  \frac{1}{2}\pounds_{n} \gamma  + \ellc \otimes_s d \ntwo.
%	\label{threetensors}                   
%\end{align}
Note that $\bU$ is symmetric and $\bF$ is a $2$-form, so in particular $\bsone(n)=0$.\ Moreover, $\bsone$, $\bU$ 
%They 
satisfy 
\cite{mars2020hypersurface}

\vspace{-0.6cm}

%\begin{small}
\noindent
\begin{minipage}[t]{0.4\textwidth}
	\begin{align}
		\pounds_{n} \ellc & = 2 \bm{\sone}, \label{soneprop}
	\end{align}
\end{minipage}%
\hfill
%\vrule
\hfill
\begin{minipage}[t]{0.6\textwidth}
	\begin{align}
		\bU (n, \cdot ) & = 0. \label{Un} 
	\end{align}
\end{minipage}

A null metric hypersurface data set $\metdata$ does not endow $\N$ with a metric tensor, so one cannot define a Levi-Civita covariant derivative.\ However, the two conditions 

\vspace{-0.55cm}

\noindent
\begin{minipage}[t]{0.4\textwidth}
	\begin{align}
		\nablao_{a} \gamma_{bc} & = - \ell_b \U_{ac} - \ell_c \U_{ab}, \label{nablaogamma} 
	\end{align}
\end{minipage}%
\hfill
\hfill
\begin{minipage}[t]{0.6\textwidth}
	\begin{align}
		\nablao_a \ell_b & = \F_{ab} - \elltwo \U_{ab} \label{nablaoll}.
	\end{align}
\end{minipage}

define a unique torsion-free covariant derivative $\nablao$ on $\N$ \cite[Prop. 4.3]{mars2020hypersurface}, called
\textit{metric hypersurface connection}.\  
The $\nablao$ derivatives of the tensor fields 
$n$ and
$P$ are \cite{mars2020hypersurface}

\vspace{-0.48cm}

\noindent
\begin{minipage}[t]{0.4\textwidth}
\begin{align}
	\nablao_b n^c &   =n^c \sone_b + P^{ac} \U_{ab}, \label{nablaonnull}
%	\label{nablaoP}\nablao_a P^{b c}&=-\left(n^b P^{c f}+n^c P^{b f}\right) \F_{a f}-n^b n^c \nablao_a \ell^{(2)}.
\end{align}
\end{minipage}%
\hfill
\hfill
\begin{minipage}[t]{0.6\textwidth}
\begin{align}
%	\nablao_b n^c &   =n^c \sone_b + P^{ac} \U_{ab}, \label{nablaonnull}\\
	\label{nablaoP}\nablao_a P^{b c}&=-\left(n^b P^{c f}+n^c P^{b f}\right) \F_{a f}-n^b n^c \nablao_a \ell^{(2)}.
\end{align}
\end{minipage}

\vspace{-0.1cm}

Later we shall need the following  identities 
%In several computations below  we shall need to 
relating $\nablao$ derivatives %to 
%derivatives
to Lie 
%derivatives 
and exterior derivatives. 
%We collect here the result that we need.
\begin{lemma}
	\label{CodazziLie}
Consider null metric hypersurface data $\metdata$ and let $\theta_a$, $S_{ab}$, $A_{ab}$ be tensor fields on $\N$ with the symmetries $S_{ab} = S_{(ab)}$ and $A_{ab} = A_{[ab]}$.\  Then,
	\begin{align}
	\label{contrNantisym}  n^b\lp \nablao_{b}\theta_{d}-\nablao_{d}\theta_{b}\rp=&\spc\pounds_{n}\theta_d  -    \nablao_{d}(\bs{\theta}({n})),\\
	n^b\lp \nablao_{b}\theta_{d}+\nablao_{d}\theta_{b}\rp=&\spc\pounds_{n}\theta_d  +   \nablao_{d}(\bs{\theta}({n}))  -  2\lp \bs{\theta}(n) \sone_d       +   P^{ab}\theta_{b}\U_{ad}   \rp,\label{contrNsym}\\
	n^c \left ( \nablao_d S_{cb} - \nablao_c S_{db} \right ) =&\spc 
	\nablao_d \left ( S_{bc} n^c \right ) 
	- \pounds_{{n}} S_{bd}  + S_{cd} \nablao_b n^c
	\label{CodazziToLie} \\
	n^c \left ( \nablao_d A_{cb} - \nablao_c A_{db} \right ) =&\spc
	\nablao_{(b} \mathfrak{a}_{d)}
	- A_{c(b} \nablao_{d)} n^c
	+ \frac{1}{2} n^c (dA)_{dcb} 
	-\frac{1}{2} n^c \nablao_c A_{db}
	\label{Codazzitos} \\
	n^c \left ( \nablao_d A_{cb} - \nablao_c A_{db} \right ) =&\spc
	n^c (dA)_{dcb} + \nablao_b \mathfrak{a}_d
	- A_{cd} \nablao_b n^c \label{Codazzitos2}
\end{align}
where $dA_{abc} \defi  3 \nablao_{[a} A_{bc]}$ and $\mathfrak{a}_{a} \defi  n^c A_{ca}$.
\end{lemma}
\begin{proof}
Since $\nablao$ is torsion-free, 
%has vanishing torsion, 
the Lie derivative of any $(0,p)$-tensor
	$T$ along any vector $V$ is %can be written as
	\begin{equation}
		\label{lieCOVtensor}(\pounds_V T)_{a_1 \cdots a_p}=V^b\nablao_b T_{a_1 \cdots a_p}+\sum_{\mathfrak{i}=1}^{p}T_{a_1\cdots a_{\mathfrak{i}-1}ba_{\mathfrak{i}+1}\cdots a_p}\nablao_{a_\mathfrak{i}}V^b.
	\end{equation}
	Equation \eqref{contrNantisym} follows from \eqref{lieCOVtensor} and $n^b\big(\nablao_{b}\theta_{d}-\nablao_{d}\theta_{b}\big)=n^b\nablao_{b}\theta_{d}+\theta_{b}\nablao_{d}n^b- \nablao_{d}\big(\bs{\theta}({n})\big)$, while %. Equation 
	\eqref{contrNsym} is obtained from 
	$n^b\big( \nablao_{b}\theta_{d}+\nablao_{d}\theta_{b}\big)=n^b\nablao_{b}\theta_{d}+ \nablao_{d}(\bs{\theta}({n}))-\theta_{b}\nablao_{d}n^b=\pounds_{n}\theta_d+ \nablao_{d}\big(\bs{\theta}({n})\big)-2\theta_{b}\nablao_{d}n^b$  
	after using \eqref{nablaonnull}.\ For the symmetric tensor $S_{cd}$, \eqref{lieCOVtensor} gives
	\begin{align*}
		n^c \left ( \nablao_d S_{cb} - \nablao_c S_{db} \right ) & =
		n^c \nablao_d S_{cb} + S_{cb} \nablao_d n^c + S_{dc} \nablao_b n^c
		- \pounds_{{n}} S_{bd} =
		\nablao_d \left ( S_{bc} n^c \right ) - \pounds_{{n}} S_{db}
		+ S_{cd} \nablao_b n^c,
	\end{align*}
	which is (\ref{CodazziToLie}).\ For $A_{ab} = A_{[ab]}$, we use $(dA)_{dcb}\defi 3\nablao_{[d}A_{cb]}=\nablao_{d} A_{cb} - \nablao_b A_{cd} - \nablao_c A_{db}$ and find
	\begin{align*}
		n^c \left ( \nablao_d A_{cb} - \nablao_c A_{db} \right )
		=&\spc \frac{1}{2} n^c \left ( \nablao_d A_{cb} + \nablao_b A_{cd}
		+ \nablao_{d} A_{cb} - \nablao_b A_{cd} - \nablao_c A_{db}  -\nablao_c A_{db}  \right ) \\
		=&\spc \frac{1}{2} \Big(  \nablao_d (A_{cb} n^c) + \nablao_b ( A_{cd} n^c)
		- A_{cb} \nablao_d n^c - A_{cd} \nablao_b n^c + n^c (dA)_{dcb} - n^c \nablao_c A_{db} \Big)
		\\
		= &\spc \nablao_{(b} \mathfrak{a}_{d)}
		- A_{c(b} \nablao_{d)} n^c
		+ \frac{1}{2} n^c (dA)_{dcb} - \frac{1}{2} n^c \nablao_c A_{db} ,
	\end{align*}
	and \eqref{Codazzitos} is established.\ Finally, to derive the alternative form
	\eqref{Codazzitos2} it suffices to notice that 
	$n^c( \nablao_d A_{cb} - \nablao_c A_{db}) 
	\hspace{-0.05cm}=\hspace{-0.05cm} n^c ( \nablao_d A_{cb} + \nablao_c A_{bd} + \nablao_b A_{dc}
	+ \nablao_b A_{cd} )  \hspace{-0.05cm}=\hspace{-0.05cm} n^c (dA)_{dcb} + \nablao_b \left ( n^c A_{cd} \right ) -
	A_{cd} \nablao_b n^c$.
%	\begin{align*}
%	\hspace{-0.2cm}n^c( \nablao_d A_{cb} - \nablao_c A_{db})
%		& 
%		\hspace{-0.05cm}=\hspace{-0.05cm} n^c ( \nablao_d A_{cb} + \nablao_c A_{bd} + \nablao_b A_{dc}
%		+ \nablao_b A_{cd} )  \hspace{-0.05cm}=\hspace{-0.05cm} n^c (dA)_{dcb} + \nablao_b \left ( n^c A_{cd} \right ) -
%		A_{cd} \nablao_b n^c.
%	\end{align*}
\end{proof}
The curvature of $\nablao$ is key to understand the geometry of null metric hypersurface data.\ We use the notation 
%$\Riemofull$, $\Ricco$ except when using abstract index notation where we just write
$\Riemo{}^{a}{}_{bcd}$, $\Riemo_{ab}$ (or  $\Riemofull$,  $\Ricco$ in index-free notation) for the curvature and Ricci tensors of $\nablao$.\ 
%we write $\Riemofull$ and  $\Ricco$ respectively. 
%Since the connection $\nablao$ is not metric,
The tensor $\Riemo_{ab}$
%$\Ricco$ 
is not symmetric in general, and in fact  it verifies %satisfies
\cite{mars2020hypersurface} %[Prop. 5.1]
%\begin{align}
%	\Riemo_{[ab]} - \Riemo_{ba} = \nablao_{[a} \sone_{b]} - \frac{1}{2} \nablao_{[a} \ntwo \nabla_{b]} \elltwo. 
%	%\\  
%	%  \color{blue}
%	%   \Ricco(X,Z) - \Ricco(Z,X) =  \left ( d \bm{\sone} -
%	% \frac{1}{2} d \ntwo \wedge d \elltwo \right )(X,Z),
%	%\quad  X,Z \in \Gamma(T\mathcal{N}).
%	\label{antisymRiemo}
%\end{align}
\begin{align}
\Riemo_{[ab]} = \nablao_{[a} \sone_{b]} . 
	%\\  
	%  \color{blue}
	%   \Ricco(X,Z) - \Ricco(Z,X) =  \left ( d \bm{\sone} -
	% \frac{1}{2} d \ntwo \wedge d \elltwo \right )(X,Z),
	%\quad  X,Z \in \Gamma(T\mathcal{N}).
\label{antisymRiemo}
\end{align}
The link between the hypersurface data formalism and the actual geometry of embedded hypersurfaces  relies on the notions of  \textit{rigging} and \textit{embeddedness} of the data.\ 
%, as follows.
A null metric data $\metdata$ is $\{\phi,\rig\}$-embedded in a Lorentzian manifold $\lp \mathcal{M}^{\n+1},g\rp$ if there exists an embedding $\phi :\mathcal{N}\longhookrightarrow\mathcal{M}$ and a rigging $\zeta$ (i.e.\ a vector field along $\phi \lp\mathcal{N}\rp$, everywhere transverse to it),  satisfying  
\begin{equation}
	\label{emhd}
	\phi ^{\star}\lp g\rp= \gamma , \qquad\phi ^{\star}\lp g\lp\zeta,\cdot\rp\rp=\ellc, \qquad\phi ^{\star}\lp g\lp\zeta,\zeta\rp\rp=\elltwo.
\end{equation}
%\begin{definition}\label{defEMHD}
%	\textup{\cite{mars2020hypersurface}}
%	%(Rigging, embedded metric hypersurface data)
%	Let $\metdata$ be a metric hypersurface data set 
%	and $\mathfrak{n}$ the dimension of $\N$. We say that the data is 
%	embedded in  $\lp \mathcal{M}^{\mathfrak{n}+1},g\rp$  provided there exists an embedding $\phi :\mathcal{N}\longhookrightarrow\mathcal{M}$ and a rigging vector field $\zeta$ (i.e.\ a vector field along $\phi \lp\mathcal{N}\rp$, everywhere transversal to it) satisfying  
%	\begin{equation}
%		\label{emhd}
%		\phi ^{\star}\lp g\rp= \gamma , \qquad\phi ^{\star}\lp g\lp\zeta,\cdot\rp\rp=\ellc, \qquad\phi ^{\star}\lp g\lp\zeta,\zeta\rp\rp=\elltwo.
%	\end{equation}
%\end{definition}
For hypersurface data $\hypdata$ embeddedness requires, in addition,
	\begin{equation}
	\label{YtensorEmbDef}
	\dfrac{1}{2}\phi^{\star}\lp \pounds_{\zeta}g\rp=\bY.
\end{equation}
%\begin{notation}
When there is no risk of confusion we shall identify scalars and vectors on $\N$ with their counterparts on $\phi(\N)$.\ The word ``abstract" will be used to refer to mathematical objects defined solely in terms of hypersurface data, e.g.\ $\gamma$, $\bF$, $\Q$ or the manifold $\N$.\ The idea is that abstract quantities can be defined irrespectively of whether the data is embedded in an ambient space.   
%\end{notation}

When the data $\hypdata$ is $\{\phi,\rig\}$-embedded in $(\M,g)$, 
the connection $\nablao$ and the Levi-Civita derivative $\nabla$ of $g$ are related by
%, given by the following  Gauss-type equation \cite{mars2013constraint}
\begin{align}
	\nabla_{X}W
	\label{nablaXYnablao}&=\phi_{\star}\Big(\nablao_{X}W-\bY(X,W)
	n\Big) - \bU(X,W)\rig, \quad\quad\forall X,W\in\Gamma(T\N).
\end{align}
%where $\nabla$ is the Levi-Civita connection of the ambient space $(\mathcal{M},g)$.

From now on, whenever we consider embedded data we will use the following setup and notation.
%we shall use in the embedded case is as follows.
\begin{setup}\label{setup:basis:e_a}
Consider null hypersurface data  $\hypdata$ $\{\phi,\rig\}$-embedded in $(\mathcal{M}, g)$.\ For any local basis $\{ \hate_a \}$ of $\Gamma(T\mathcal{N})$, 
%and define $e_a \defi \phi _{\star} (\hate_a)$.\ By transversality of the rigging, 
 $\{\zeta,e_a \defi \phi _{\star} (\hate_a)\}$ is a 
	%(local) 
	basis of $\Gamma(T\mathcal{M})\vert_{\phi (\mathcal{N})}$.\ Let $\bs{\nu}$ be the	unique normal covector satisfying $\bs{\nu}(\zeta) = 1$.\ We denote the dual basis of $\{\zeta,e_a\}$ by $\{\bs{\nu}, \bs{\theta}^a\}$, and define 
	the vector fields 
	$\nu \defi g^{\sharp}(\bs{\nu},\cdot)$, ${\theta}^a \defi g^{\sharp}(\bs{\theta}^a,\cdot)$.
\end{setup}
From \eqref{prod1}-\eqref{prod4} one easily checks that $\nu$ and $\theta^a$ decompose as
%The definition of dual basis combined with \eqref{prod1}-\eqref{prod4} implies the following decomposition of the vectors $\nu$ and $\theta^a$

\vspace{-0.27cm}

\begin{multicols}{2}
	\noindent
	\begin{align}
		\label{normal}
		\nu & =  n^a e_a, \\
		\theta^a & = P^{ab} e_b + n^a \rig,  \label{omegas}
	\end{align}
\end{multicols}

\vspace{-0.4cm}

%in the basis $\{\rig,e_a\}$, 
so $\nu= \phi_{\star} n$ in the embedded case.\ 
One can also prove that 
$\bU$ coincides \cite{mars2020hypersurface} with the second fundamental form of $\phi(\N)$ w.r.t.\ the normal vector $\nu$.\ %Moreover,
%Moreover, 
By  \eqref{emhd} the components of $\bs{\A}$ in the basis $\{(\hate_a,0),(0,1)\}$ coincide with those of $g$ in the basis $\{\rig,e_a\}$, 
%.\ This allows one to write %, 
so 
%To express 
the contravariant metric $g^{\sharp}$ is written 
%can be written 
in terms of 
%the basis 
$\{\rig,e_a\}$ as 
%by making use of 
%,  note  that
(cf.\ \eqref{ambientinversemetric}) 
\begin{align}
	\label{gup} g^{\mu\nu} & \stackbin{\phi (\mathcal{N})}{=} 
	 n^c \left ( \rig^{\mu} e_c^{\nu} +
	\rig^{\nu}e_c^{\mu}  \right ) +P^{cd} e_c^{\mu} e_d^{\nu}
	\qquad \Longleftrightarrow\qquad g^{\mu\nu}\stackbin{\phi (\mathcal{N})}{=} e_c^{\mu} \theta^c{}^{\nu} + \rig^{\mu} \nu^\nu,
\end{align}
where the equivalence follows from  \eqref{normal}-\eqref{omegas}.\ In particular, from \eqref{gup} it follows that the square norm of 
%the normal 
$\bs{\nu}$ vanishes everywhere on $\N$, which is consistent with $\nu$ being null at the hypersurface.\

Further on we shall need the tangential derivative of the rigging vector $\rig$, given by \cite{manzano2023field} 
\begin{align}
	\nabla_{e_a} \rig =\frac{1}{2} \nablao_a \elltwo \nu
	+ \left ( \Y_{ab} + \F_{ab} \right ) \theta^b.
	\label{nablarig3}
\end{align}
%A direct way of proving this is by 
%For its later use, we include the 
It will also be important below to have the explicit form of those
components of the curvature tensor $R_{\alpha\beta\gamma\delta}$ of $(\mathcal{M},g)$ that are computable in terms of  hypersurface data \cite{mars2013constraint, mars1993geometry}, namely 
\begin{align}
	\hspace{-0.2cm}R_{\mu\alpha\nu\beta} \rig^{\mu} e^{\alpha}_b e^{\nu}_c e^{\beta}_d =& \spc\ell_a \Riemo{}^a_{\,\,\,bcd} + 2 \nablao_{[d} {\Y}_{c]b} 
	+ 2 \elltwo \nablao_{[d} \U_{c]b}+  \U_{b[c} \nablao_{d]} \elltwo   + 2{\Y}_{b[d}  
	\left ( \F_{c]f} + {\Y}_{c]f} \right ) n^f  , \label{Codazzi} \\
	\hspace{-0.2cm}R_{\mu\alpha\nu\beta}e^{\mu}_a e^{\alpha}_b e^{\nu}_ce^{\beta}_d 
	=& \spc 
	\gamma_{af} \Riemo{}^f_{\,\,\,bcd}
	+ 2 \ell_a \nablao_{[d} \U_{c]b}
	+ 2 {\Y}_{b[c} \U_{d]a}  + 2 \U_{b[c} {\Y}_{d]a}  + 2 \U_{b[c} \F_{d]a}. \label{Gauss}
\end{align}
%\end{proposition}

\subsection{Gauge structure}\label{sec:gauge:structure:prelim}

Rigging vectors on embedded hypersurfaces are highly non-unique.\ In the hypersurface data formalism, this fact is captured by a built-in gauge freedom  (see \cite{mars2013constraint, mars2020hypersurface} for details).\ 
Consider null hypersurface data $\hypdata$, a no-where zero function $z\in\mathcal{F}^{\star}(\mathcal{N})$ and a vector field $V\in\Gamma(T\mathcal{N})$.\ The gauge-transformed  data 
$$\mathcal{G}_{\lp z,V\rp}(\hypdata) \defi \lb \mathcal{N},\mathcal{G}_{\lp z,V\rp}\lp \gamma \rp,\mathcal{G}_{\lp z,V\rp}\lp\ellc\rp,\mathcal{G}_{\lp z,V\rp}\big( \ell^{(2)}\big),\G_{(z,V)}\big(\bY \big)\rb$$ 
is defined  by
\begin{align}
	\label{gaugegamma&ell2} \hspace{-0.2cm}\mathcal{G}_{\lp z,V\rp}(\gamma)& \defi  \gamma ,\hspace{0.42cm}\mathcal{G}_{\lp z,V\rp}\lp\ellc\rp \defi z\lp\ellc+ \gamma \lp V,\cdot\rp\rp,\hspace{0.42cm}\mathcal{G}_{\lp z,V\rp}\big(\ell^{(2)}\big) \defi z^2\big( \ell^{(2)}+2\ellc\lp V\rp+ \gamma \lp V,V\rp\big),\\
	\label{gaugeY}\hspace{-0.2cm}\mathcal{G}_{\lp z,V\rp}(\bY)  & \defi z\bY+ \big( \ellc+\gamma(V,\cdot)\big)\otimes_s dz+\frac{z}{2}\lieo_{V} \gamma. 
	%=z\bY+ \lp \ellc+\gamma(V,\cdot)\rp\otimes_s dz+\frac{z}{2}\lieo_{V} \gamma.
\end{align} 
%
%
%\vspace{-0.6cm}
%
%
%\noindent
%\begin{minipage}[t]{0.4\textwidth}
%	\begin{align}
%		\label{gaugegamma&ell2} \mathcal{G}_{\lp z,V\rp}\lp \gamma \rp& \defi  \gamma ,
%	\end{align}
%\end{minipage}%
%\hfill
%\hfill
%\begin{minipage}[t]{0.6\textwidth}
%	\begin{align}
%		\label{gaugeell}\mathcal{G}_{\lp z,V\rp}\lp\ellc\rp& \defi z\lp\ellc+ \gamma \lp V,\cdot\rp\rp, 
%	\end{align}
%\end{minipage}
%
%\vspace{-0.2cm} 
%
%\begin{align}
%	\label{gaugeell2}\hspace{-5.7cm}\mathcal{G}_{\lp z,V\rp}\big(\ell^{(2)}\big) \defi z^2\big( \ell^{(2)}+2\ellc\lp V\rp+ \gamma \lp V,V\rp\big),\hfill\\
%	\label{gaugeY}\mathcal{G}_{\lp z,V\rp}\lp \bY\rp   \defi
%	%z\bY+ \ellc\otimes_s \tdo z+\frac{1}{2}\lieo_{zV} \gamma=
%	z\bY+ \lp \ellc+\gamma(V,\cdot)\rp\otimes_s dz+\frac{z}{2}\lieo_{V} \gamma.
%\end{align} 
%
%As proven in \cite[Lemma 3.3]{mars2020hypersurface}, 
The set of transformations $\{\G_{(z,V)}\}$ forms a group $\mathcal{G}=\Fcal(\N)\times\Gamma(T\N)$ with composition law $\mathcal{G}_{\left(z_2, V_2\right)} \circ\mathcal{G}_{\left(z_1, V_1\right)}=\mathcal{G}_{\left(z_1 z_2 , V_2+z_2^{-1} V_1\right)}$, identity  $\mathcal{G}_{\mathbb{I}} \defi \mathcal{G}_{(1,0)}$ and inverse  $\mathcal{G}^{-1}_{(z,V)} \defi \mathcal{G}_{(z^{-1},-zV)}$. We call  $\G$ the \textit{gauge group}, each element $\G_{(z,V)}$ a \textit{gauge transformation} (also \textit{gauge group element}) and the quantities $\{z,V\}$ \textit{gauge parameters}.\ 

A gauge transformation of a hypersurface data induces transformations on the rest of geometric quantities.\ In particular,  %tensor fields $P$ and $n$.\ For later use, we include the induced transformation of $n$, namely
\begin{align}
	\label{gaugen}\mathcal{G}_{\lp z,V\rp}\lp n\rp&=z^{-1}n.
\end{align}
The connection between ${\mathcal G}$ (which is intrinsic to the data) and the freedom in choosing the rigging is as follows
\cite[Prop. 3.4]{mars2020hypersurface}.\  
If $\hypdata$ is $\{\phi,\rig\}$-embedded in $(\mathcal{M},g)$, then  $\mathcal{G}_{(z,V)}(\hypdata)$ is also embedded in  $({\mathcal M},g)$ with embedding $\phi$ but 
different rigging %$\mathcal{G}_{(z,V)}(\rig)$ given by
\begin{equation}
	\label{gaugerig}\mathcal{G}_{(z,V)} (\rig)  \defi  z (\rig + \phi _{\star} V).
\end{equation}
%Note that this implies the following gauge behaviour for the
%normal vector $\nu$ 
%\begin{equation}
%	\G_{(z,V)}(\nu)=\frac{1}{z}\nu.
%\end{equation}

\section{Transverse submanifolds}\label{seclevicivitaonS}

Spacelike cross-sections are an important tool to analyze the geometry of null hypersurfaces.\ In line with this, in this section we discuss several geometric properties of codimension-one non-degenerate submanifolds embedded in hypersurface data.\ 
Complementary
results can be found in 
\cite{manzano2023PhD, mars2023covariant}.

%Given null metric hypersurface data $\metdata$,  
A \textit{transverse submanifold $S$ is a codimension-one embedded submanifold of $\N$ to which $n$ is everywhere transverse}.\ Existence of such $S$ is always guaranteed in sufficiently local domains of $\N$.\ Note that we are not assuming that $S$ is a global section of $\N$, i.e.\ there can be generators of $\N$ that do not cross $S$.\ However, we enforce that generators intersecting 
$S$ do it only once.

To analyze the geometry of $S$ we put forward the following setup.
\begin{setup}\label{setup}
%Let $\metdata$ be null metric hypersurface data and $S$ be 
Consider a transverse submanifold $S$
with embedding $\psi:S\longhookrightarrow \mathcal{N}$ in null metric hypersurface data $\metdata$.\ 
We define $\ellc_{\parallel}\defi \psi^{\star}\ellc$ and let $\bnormal$ be the only normal covector along $\psi(S)$ 
satisfying $\bnormal(n)=1$.\ We take a basis $\{\hat{v}_A\}$ of $\Gamma(TS)$ and construct the basis $\{n,v_A\defi \psi_{\star}(\hat{v}_A)\}$ of $\Gamma(T\N)\vert_{\psi(S)}$.
\end{setup}
%In Setup \ref{setup}, 
Let us prove that the induced metric $h\defi \psi^{\star}\gamma$ is non-degenerate.\ We need to show that a vector $X \in T_p S$ which is $h$-orthogonal to all $T_p  S$ is necessarily zero.\ 
%By definition, such a vector 
$X$ satisfies that $\psi_{\star}\vert_p (X)$ is $\gamma$-orthogonal to all $T_p\N$ because $T_p \N = T_p S  \oplus \langle n\vert_p\rangle$  and $\gamma(n, \cdot)\vert_p =0$.\  Thus, $\psi_{\star}\vert_p (X) 
\in \text{Rad} (\gamma\vert_p)$ and hence it must be proportional to $n\vert_p$. Since
$\psi_{\star}\vert_p (X)$  is also tangent to $\psi(S)$, this can only occur if $X =0$.

We let $h^{\sharp}$ denote the contravariant metric of $h$.\ In index notation we write $h_{IJ}$ and $h^{IJ}$, and use these objects to lower and raise Capital Latin indices.\ 
We introduce the vector $\ell_{\parallel}\defi h^{\sharp}(\ellc_{\parallel},\cdot)$ (with components $\ell^A$) and the scalar $\elltwo_{\parallel}\defi h^{\sharp}(\ellc_{\parallel},\ellc_{\parallel})$.\  We will frequently simplify the notation and identify $S$, $X\in\Gamma(TS)$, $f\in\Fcal(\psi(S))$ with their counterparts $\psi(S)$, $\psi_{\star}X$ and $\psi^{\star}f$.\  Given any $(0,p)$-tensor $\bs{T}$ along $S$, we define $\bs{T}_{\parallel}\defi \psi^{\star}\bs{T}$ and write $T_{A_1\dots A_{p}}\defi \bs{T}_{\parallel}(\hat{v}_{A_1},\dots, \hat{v}_{A_{p}})$ for its components (note that we
{\em do not} write the parallel symbol in the components).

For later use, we next provide the decompositions of $\gamma$ and $P$ in a basis adapted to $S$.\ %We do this in the next lemma. %, together with various useful identities
\begin{lemma}\label{corPUandPUU}
	In the Setup \ref{setup}, let  $\{\bnormal,\bs{\theta}^A\}\in\Gamma(T^{\star}\psi(S))$ be the dual basis of $\{n,v_A\}$.\   
	Then, %the tensors $\gamma$ and $P$ decompose as
	\begin{equation}
		%\label{gammadecom:abstract}
		\gamma=h_{AB}\bm{\theta}^A\otimes\bm{\theta}^B,\qquad
		P=h^{AB}v_A\otimes v_B-2\ell^{A} n\otimes_s v_A + \left(\elltwo_{\parallel}- \elltwo \right) n\otimes n, \label{Pdecom:abstract}
	\end{equation}
	%In the Setup \ref{setup}, 
	and the following identities hold: 
	\begin{align}
		\label{PUandPUU}P^{cf}\U_{fa}&=h^{IJ}v_J^f(v_I^c-\ellp_In^c)\U_{fa},& P^{cd}\U_{ac}\U_{bd}&=h^{IJ}v_I^cv_J^d\U_{ac}\U_{bd},\\
		\label{trPYandtrPUwithPdec}\textup{tr}_P\bY&=\textup{tr}_h\bYp-2\ellp^Ar_A+\Q(\elltwo-\elltwo_{\parallel}), & \textup{tr}_P\bU&=\textup{tr}_h\bUp.
	\end{align}
\end{lemma}
\begin{proof}
	The first expression in \eqref{Pdecom:abstract} is a consequence of $\gamma(n,\cdot)=0$ and $h=\psi^{\star}\gamma$.\ On the other hand, since $P$ is symmetric it decomposes in the basis $\{n,v_A\}$ as
	\begin{align}
		\label{Pdecom:proof}P&=P(\bs{\theta}^A,\bs{\theta}^B)v_A\otimes v_B+P(\bnormal,\bs{\theta}^A) (n\otimes v_A+v_A\otimes n)+P(\bnormal,\bnormal) n\otimes n.
	\end{align}
	We now use that 
	%the fact that 
	%The covector 
	%$\ellc$ decomposes 
	%in the basis $\{\bnormal,\bs{\theta}^A\}$ 
	%as 
	$\ellc=\bnormal+\ell_A\bs{\theta}^A$ (because $\ellc(n)=1$ (cf.\ \eqref{prod2}) and $\ell_A = \ellc(v_A)$) and compute
	\begin{align*}
		P(\bnormal,\cdot)& =P(\ellc-\ell_A\bs{\theta}^A,\cdot)\stackbin{\eqref{prod3}}=-\elltwo n-\ell_A P(\bs{\theta}^A,\cdot)=-\lp\elltwo+\ell_A P(\bs{\theta}^A,\bnormal)\rp n-h^{AB}\ell_Av_B,\\
		\delta_A^B & =\delta_a^b\theta_b^Bv_A^a\stackbin{\eqref{prod4}}=(P^{bf}\gamma_{fa}+n^b\ell_a)\theta_b^Bv_A^a=P^{bf}\gamma_{fa}\theta_b^Bv_A^a
		%\stackbin{\eqref{gammadecom:abstract}}
		=h_{AC}P(\bs{\theta}^B,\bs{\theta}^C).
	\end{align*}
	%where in the second line we .\ 
	It follows that $$P(\bs{\theta}^A,\bs{\theta}^B)=h^{AB},\quad P(\bnormal,\bs{\theta}^C)=-h^{AC}\ell_A=-\ell^C,\quad P(\bnormal,\bnormal)=-\lp\elltwo- h^{AB}\ell_A\ell_B\rp= \elltwo_{\parallel}-\elltwo,$$ which establishes the second expression in  \eqref{Pdecom:abstract}.\ 
	%	\begin{align}
		%		\color{red}
		%			\label{Pdecomp2}P^{cf}&=h^{AB}v_A^cv_B^f-h^{AB}\ellp_B(n^cv^f_A+n^fv_A^c)+(\elltwo_{\parallel}-\elltwo)n^cn^f
		%		\end{align} 
	Identities \eqref{PUandPUU} are then an immediate consequence of the decomposition of $P$ in \eqref{Pdecom:abstract}, as $\bU(n,\cdot)=0$ (cf.\ \eqref{Un}).\ Finally, for $\textup{tr}_P\bY$ and $\text{tr}_p\bU$ we find
	%Inserting \eqref{Pdecom:abstract} in the definition of  $\textup{tr}_P\bY$ yields
	\begin{align*}
		\textup{tr}_P\bY
		&= P^{cd}\Y_{cd}=(h^{CD}v_C^cv_D^d-\ellp^D(n^cv_D^d+n^dv_D^c)-(\elltwo-\elltwo_{\parallel})n^cn^d)\Y_{cd},\\
		%= \textup{tr}_h\bYp-2\ellp^Dr_D+\Q(\elltwo-\elltwo_{\parallel}),
		\textup{tr}_P\bU
		&= P^{cd}\U_{cd}
		%=(h^{CD}v_C^cv_D^d-\ellp^D(n^cv_D^d+n^dv_D^c)-(\elltwo-\elltwo_{\parallel})n^cn^d)\U_{cd}
		\stackbin{\eqref{Un}}=h^{CD}v_C^cv_D^d\U_{cd}
		=h^{CD}\U_{CD},
	\end{align*}
	from where \eqref{trPYandtrPUwithPdec} follows 
	%which becomes the first expression in \eqref{trPYandtrPUwithPdec} 
	after using the definitions \eqref{defY(n,.)andQ}.\ 
\end{proof}
The main results in this section are  
the relation between the connection $\nablao$ and the Levi-Civita  derivative $\nabla^h$ on $S$ and    
an expression 
%, is the relation between 
%This will allow us to relate 
for the tangential components of the curvature tensor of $\nablao$ in terms of the curvature tensor of 
%the induced metric 
$h$.\ 
To establish the former, 
%Let us now establish the relationship between the covariant derivatives $\nablao$ and $\nabla^h$.\ For that 
we need to use that 
% A key ingredient for the result is the fact that 
%that 
the direct sum $T_p\N=T_pS\oplus \langle n\vert_p\rangle$ induces the following unique decomposition for any vector field $X$ along $\psi(S)$:  
%can be decomposed uniquely as
\begin{align}
	X = X^{\parallel} + f_X n,\qquad\text{where}\qquad X^{\parallel}\in\Gamma(TS),\quad f_X\in\Fcal(S).  \label{decomX}
\end{align}
%where $f_X$ is a function on $S$ and $X^{\parallel}$ is tangent to $S$.\ 
%This  follows from the fact that  $n$ is transverse to $\psi(S)$ at all its points.
\begin{lemma}\label{lem:nablao:nablah}
	%\color{blue}
	In the Setup \ref{setup}, let $\nabla^h$ be the Levi-Civita 
	covariant 
	derivative of $h$.\  %on $(S,h)$ 
	%and $\elltwo_{\parallel}\defi h^{\sharp}(\ellc_{\parallel},\ellc_{\parallel})$. 
	Then, for any $X,Z\in\Gamma(TS)$,  it holds
	\begin{align}
		\nn \nablao_XZ=&\spc\nabla^{h}_XZ+h^{\sharp}(\ellc_{\parallel},\cdot)\bUp(X,Z)\\
		\label{nABLAwithnablah} &+\Bigg(\dfrac{1}{2}\lp \big(\nabla^h_X\ellc_{\parallel}\big)\lp Z\rp+\big(\nabla^h_Z\ellc_{\parallel}\big)\lp X\rp\rp+\big(\elltwo -\elltwo_{\parallel}\big)\bUp(X,Z)\Bigg) n.
	\end{align}
		
	\vspace{-0.4cm}
	
\end{lemma}
\begin{proof}
	%\color{blue}
	The difference of two connections is a tensor.\ Using the decomposition \eqref{decomX} we can write
	\begin{equation}
		\label{new:relation:nablao:nabh}
		\nablao_XZ=\nabh_XZ-\Xi(X,Z)+\Omega(X,Z)n,
	\end{equation}
	where $\Xi$, $\Omega$ are a $(1,2)$-tensor and 
	%$\Omega$ is 
	a $(0,2)$-tensor on $S$.\ 
	Since both $\nablao$ and $\nabh$ are torsion-free, $\Xi(X,Z)$ and $\Omega(X,Z)$ are symmetric in $X,Z$.\  
	To determine the explicit form of $\Xi$, we consider another vector $W$ tangent to $S$,
	and find %proceed as follows:
	\begin{align}
		\nonumber 0&=(\nabh_Xh)(Z,W)=\nabh_X(h(Z,W))-h(\nabh_XZ,W)-h(\nabh_XW,Z)\\
		\nn &=\nablao_X(\gamma(Z,W))-\gamma(\nablao_XZ+\Xi(X,Z),W)-\gamma(\nablao_XW+\Xi(X,W),Z)\\
		\nn &=(\nablao_X\gamma)(Z,W)-h(\Xi(X,Z),W)-h(\Xi(X,W),Z)\\
		\label{middleeq18} &=- \ellc_{\parallel}(Z) \bUp(X,W) - \ellc_{\parallel}(W) \bUp(X,Z)-h(\Xi(X,Z),W)-h(\Xi(X,W),Z),
	\end{align}
	where in the third equality we have used \eqref{new:relation:nablao:nabh} and $\gamma(n, \cdot)=0$, on the fourth that $\Xi(X,Z)\in\Gamma(TS)$ and in the last one we inserted \eqref{nablaogamma}. In index notation, \eqref{middleeq18} reads 
	\begin{equation}
		-\ell_B\bU_{AC}-\ell_C\U_{AB}-h_{CD}\Xi^{D}_{AB}-h_{BD}\Xi^{D}_{AC}=0,
	\end{equation}
	from which it is straightforward to prove that  
	${\Xi}^A_{BC}=- h^{AD}\ell_D\U_{BC}$ by adding three copies of the expression after applying a standard index permutation.\  
	To determine $\Omega$ we use its symmetry to rewrite
	\eqref{new:relation:nablao:nabh} as
	\begin{align*}
		2\Omega(X,Z)n=\nablao_XZ-\nabh_XZ+\nablao_ZX-\nabh_ZX+2\Xi(X,Z)\spc\spc.
	\end{align*}
	Contracting with the normal covector $\bnormal$ yields %implies
	\begin{align}
		\label{middleeq16}
		2\Omega(X,Z)=\bnormal\big(\nablao_XZ+\nablao_ZX\big).
	\end{align}
	The right-hand side can be rewritten
	in terms of $\bUp\defi \psi^{\star}\bU$ and derivatives of $\ellc_{\parallel}$  as follows.\ 
	Note first that for any vector field $X$ along $\psi(S)$ we have
	$\bnormal(X)=\ellc(X) - \ellc_{\parallel}(X^{\parallel})$ (if $X = X^{\parallel}$ it is true because $\ellc_{\parallel} = \psi^{\star} \ellc$; if
	$X = f_X n$ it is true because $\bnormal(n) = \ellc(n)=1$).\ 
	Using this and 
	\eqref{new:relation:nablao:nabh} we get 
	\begin{align}
		\nn\bnormal(\nablao_XZ)=&\spc\ellc(\nablao_XZ)-\ellc_{\parallel}(\nabh_XZ-\Xi(X,Z))=X\lp\ellc\lp Z\rp\rp-(\nablao_X\ellc)(Z)-\ellc_{\parallel}(\nabh_XZ-\Xi(X,Z))\\
		\nn =&\spc X(\ellc_{\parallel}(Z))-\bF(X,Z) + \elltwo \bUp(X,Z)-\ellc_{\parallel}(\nabh_XZ)+\ellc_{\parallel}\big(\Xi(X,Z)\big)\\
		=&\spc (\nabh_X\ellc_{\parallel})\lp Z\rp-\bF(X,Z) + \elltwo \bUp(X,Z)+\ellc_{\parallel}\big(\Xi(X,Z)\big),
	\end{align}
	where in the third equality we used \eqref{nablaoll}.
	Inserting this into \eqref{middleeq16} and using that $\bF$ is antisymmetric gives 
	\begin{align}
		\nn \Omega(X,Z)=&\spc\frac{1}{2}\lp (\nabh_X\ellc_{\parallel})\lp Z\rp+(\nabh_Z\ellc_{\parallel})\lp X\rp\rp+ \elltwo \bUp(X,Z)+\ellc_{\parallel}\big(\Xi(X,Z)\big),
		\\
		\label{Omega(X,Y)final} =&\spc\dfrac{1}{2}\lp \big(\nabla^h_X\ellc_{\parallel}\big)\lp Z\rp+\big(\nabla^h_Z\ellc_{\parallel}\big)\lp X\rp\rp+\big(\elltwo -\elltwo_{\parallel}\big)\bUp(X,Z),
	\end{align}
	which 
	completes the proof of \eqref{nABLAwithnablah}.
\end{proof}
Having obtained a Gauss-type equation relating the covariant derivatives
$\nablao$ and $\nabla^h$, we can now relate the tangential components of the curvature tensor of $\nablao$ and the curvature tensor of the induced metric $h$.\ This result relies on a generalized Gauss identity that we derive in Appendix \ref{secGauss}.\ Recall that on a semi-Riemannian manifold, the Gauss identity is an equation relating the curvature tensor of the Levi-Civita connection along tangential directions of a non-degenerate hypersurface with the curvature tensor of the induced metric and the second fundamental form. In Appendix \ref{secGauss}, we have extended this result to the more general case when the connection of the space and of the hypersurface are completely general, except for the condition that they are both torsion-free.\ By particularizing Theorem \ref{thmGeneralIdentity} (more specifically its abstract index notation form \eqref{gammaRiemanvvvGENERAL}) to the case of null hypersurface data, we get to the following result.
\begin{lemma}
	Consider null metric hypersurface data $\metdata$ and assume Setup \ref{setup}.\ Let $R^h_{ABCD}$ the Riemann tensor of $(S,h)$. Then,
	\begin{align}
		\nn v_A^a\gamma_{af}\Riemo{}^f{}_{bcd}& v_B^bv_C^cv_D^d= R^h_{ABCD}+2\nabla^h_{[C\vert}(\ellp_A\U_{B\vert D]})+\ellp_A\ellp^F\lp \U_{BD}\U_{CF}-\U_{BC}\U_{DF}\rp\\
		\label{gammaRiemanvvvOURCASE} &+\U_{AC}\lp (\elltwo-\elltwo_{\parallel})\U_{BD}+\nabla^h_{(B}\ellp_{D)}\rp-\U_{AD}\lp (\elltwo-\elltwo_{\parallel})\U_{BC}+\nabla^h_{(B}\ellp_{C)}\rp.
	\end{align}
\end{lemma}
\begin{proof}
	We particularize Theorem \ref{thmGeneralIdentity} for 
	%$\M=\N$, 
	$\whnabla=\nablao$, $\whD=\nabla^h$, $\widehat{\gamma}=\gamma$.\ Then, $\widehat{h}=h$ and \eqref{nABLAwithnablah}-\eqref{Omega(X,Y)final} hold, which means that $A{}^C{}_{AB}=\ellp^C\U_{AB}$, $A_h{}_{CAB}=\ellp_C\U_{AB}$ and $\Omega_{AB}=\nabla^h_{(A}\ellp_{B)}+(\elltwo-\elltwo_{\parallel})\U_{AB}$.\ The only term that needs further evaluation is
	$v_D^dv_A^a(\nablao_d\gamma_{af})\mathcal{P}^{f}{}_{BC}$.\ This is straightforward from \eqref{prod2} and \eqref{nablaogamma}, namely
	\begin{align}
		\nn v_D^dv_A^a(\nablao_d\gamma_{af})\mathcal{P}^{f}{}_{BC}=&\spc-v_D^dv_A^a(\ell_a\U_{df}+\ell_f\U_{da})(v_F^fA^{F}{}_{BC}+n^f\Omega_{BC})\\
		\nn =&\spc-\ellp_A\ellp^F\U_{DF}\U_{BC}-\elltwo_{\parallel}\U_{DA}\U_{BC}-\U_{DA}\Omega_{BC}\\
		\label{mideq23} =&\spc-\ellp_A\ellp^F\U_{DF}\U_{BC}-\U_{DA}(\nabla^h_{(B}\ellp_{C)}+\elltwo\U_{BC}).
	\end{align}
	Equation \eqref{gammaRiemanvvvOURCASE} follows at once after inserting \eqref{mideq23} into \eqref{gammaRiemanvvvGENERAL} and using $\gamma(n,n)=0$.
\end{proof}
We conclude the section 
with two identities involving pull-backs to $S$ of derivatives of covariant tensors.\  
%with two results needed later.\ 
We first compute the pull-back of the $\nablao$-derivative of a general $(0, p)$-tensor field $\mathcal{T}$, and then find the pull-back of the Lie derivative 
along any direction
of a general symmetric $(0,2)$-tensor $\bs{T}$ satisfying $\bs{T}(n,\cdot)=0$. 
\begin{lemma}\label{lem:pullbacktoS}
	In the Setup \ref{setup}, let $\mathcal{T}$ be any $(0, p)$-tensor field along $S$ and define %where 
	$\mathcal{T}_{\parallel}\defi \psi^{\star}\mathcal{T}$.\ Then,
	\begin{align}
		\nn &v_{A_1}^{a_1}\dots v_{A_{ p}}^{a_{ p}}v_B^b\nablao_b\mathcal{T}_{a_1\cdots a_{ p}} = \nabla_{B}^h\mathcal{T}_{A_1\cdots A_{ p}}-\sum_{\mathfrak{i}=1}^{ p}\ellp^J\mathcal{T}_{A_1\cdots A_{\mathfrak{i}-1}JA_{\mathfrak{i}+1}\cdots A_{ p}} \U_{A_{\mathfrak{i}}B}\\
		\label{covderpcovtensoronS}  &-\sum_{\mathfrak{i}=1}^{ p}\mathcal{T}_{a_1\cdots a_{ p}}v_{A_1}^{a_1}\dots v_{A_{\mathfrak{i}-1}}^{a_{\mathfrak{i}-1}}n^{a_{\mathfrak{i}}}v_{A_{\mathfrak{i}+1}}^{a_{\mathfrak{i}+1}}\dots v_{A_{ p}}^{a_{ p}}\lp \nabla^h_{(A_{\mathfrak{i}}}\ellp_{B)}+(\elltwo-\elltwo_{\parallel})\U_{A_{\mathfrak{i}}B} \rp.
	\end{align}
	%where $\mathcal{T}_{\parallel}\defi \psi^{\star}\mathcal{T}$.
\end{lemma}
\begin{proof}
	We prove it only for covectors. The case of covariant tensors with more indices is analogous. Combining \eqref{nABLAwithnablah} and \eqref{Omega(X,Y)final}, we obtain 
	\begin{align}
		\nn v_{A}^{a}v_B^b\nablao_b\mathcal{T}_{a}=&\spc v_B\lp \mathcal{T}_{A}\rp-\mathcal{T}_{a}v_B^b\nablao_b v_{A}^{a}= v_B\lp \mathcal{T}_{A}\rp-\mathcal{T}_J(\nabla^h_{v_B}v_{A}^{J}+\ellp^J \U_{AB})-\mathcal{T}_{a}n^{a}\Omega_{AB}\\
		\nn =&\spc \nabla^h_{B} \mathcal{T}_{A}-\ellp^J\mathcal{T}_J \U_{AB}-\mathcal{T}_{a}n^{a}\lp \nabla^h_{(A}\ellp_{B)}+(\elltwo-\elltwo_{\parallel})\U_{AB}\rp.
	\end{align}
\end{proof}
\begin{lemma}\label{pullbackLieT}
	In the Setup \ref{setup}, let $\bT$ be a symmetric
	$(0,2)$-tensor on $\N$ satisfying $\bT(n, \cdot )=0$. Consider a smooth function $q\in\Fcal(\psi(S))$ and a covector field $\bm{\beta}\in\Gamma(T^{\star}\N)\vert_{\psi(S)}$ verifying $\bm{\beta} (n)=0$, and define $t^a \defi  q n^a + P^{ab} \beta_b$.\ Then,
	\begin{align}
		\lp   \pounds_{t}  T \rp_{AB} =  (q- \ell^C \beta_C ) \vert_S (\pounds_n T)_{AB}
		+  \beta^C \nabh_C T_{AB}
		+ T_{AC} \nabh_B \beta^C + T_{CB} \nabh_A \beta^C. \label{pullbackLie}
	\end{align}
\end{lemma}
\begin{proof}
	Using the decomposition of $P^{ab}$ in  \eqref{Pdecom:abstract} 
	%in Lemma \ref{lemPdec} 
	and the fact that
	$\beta_a n^a =0$, we write
	\begin{align*}
		t^a  = q n^a + h^{AB} v_A^a v_B^b \beta_b - h^{AB} \ell_B n^a v_A^b \beta_b
		= ( q - \ell^A \beta_A ) n^a + \beta^A v_A^a.
	\end{align*}
	For any function $f$ we have $\pounds_{f n} \bT = f \pounds_n \bT$ because
	$\bT(n, \cdot) = \bT(\cdot, n)=0$. On the other hand, for any vector field $ W $ tangent to $S$ (i.e.\ such that there exists $\overline{ W } \in \Gamma(TS)$ so that $ W  |_S = \psi_{\star} \overline{ W }$) it holds $\psi^{\star} \lp \pounds_{ W } \bT \rp = \pounds_{\overline{ W }}
	\lp \psi^{\star} \bT \rp$. 
	Thus,
	\begin{align*}
		\psi^{\star} \lp   \pounds_{t}  \bT \rp =
		\psi^{\star} (\pounds_{ ( q- \ell^C \beta_C) n} \bT)
		+         \pounds_{\beta^{\sharp}} \lp \psi^{\star} \bT \rp =
		( q - \ell^C \beta_C ) \vert_S \psi^{\star} \lp \pounds_n \bT \rp
		+ \pounds_{\beta^{\sharp}} \lp \psi^{\star} \bT \rp,
	\end{align*}
	where $\beta^{\sharp}$ is the vector field in $S$ with abstract index components $\beta^A$. Since $\nabh$ is torsion-free the last term can be expanded in terms of the covariant derivative and \eqref{pullbackLie} follows.
\end{proof}

\section{The constraint tensor}
\label{secAbstractRicci}

This section is devoted to the main object of this paper, namely
the constraint tensor.\  %of a null hypersurface data set.\ 
This tensor was originally defined for null hypersurface data in \cite{mars2023covariant,Mars2023first} in terms of the so-called \textit{hypersurface connection} $\ovnabla$.\ The derivative $\ovnabla$ is another torsion-free connection that can be defined in terms of the \textit{full} hypersurface data (including the tensor field $\bY$) and which,  in the embedded case, coincides with the connection induced from the Levi-Civita covariant derivative of the ambient space \cite{mars2013constraint}.\ In \cite{mars2023covariant, Mars2023first}, the definition of the constraint tensor  is not fully explicit in the tensor $\bY$, as the connection $\ovnabla$ and its corresponding curvature
$\ov{R}$ depend on it.\ Moreover, in \cite{mars2023covariant, Mars2023first} the constraint tensor 
%in \cite{mars2023covariant,Mars2023first} was
is not expanded in terms of the  data, 
but instead  
%, as we have done here in expression \eqref{defabsRicci}.\ Instead, 
it is decomposed in terms of a foliation by spacelike hypersurfaces, in  analogy with other forms of the constraint equations that have appeared in the literature.\ 
While the setup in \cite{mars2023covariant,Mars2023first} is certainly helpful, in many circumstances it is more convenient to have a definition of the constraint tensor that 
%Definition \ref{defabsRicci}, on the other hand, 
shows its full dependence on the tensor $\bY$.\  
%(see Section \ref{sec:application} for an example of this).\ 
%(in the terms involving $\bY$, $\bs{\Yn}$ and $\Q$),
%as both $\nablao$ and $\Ricco$ depend only on the metric part of the data.\  
Such definition has been proposed in \cite{manzano2023PhD} as part of the Ph.\ D.\ thesis of the first named author of this paper.\ In that work, %Specifically, in \cite{manzano2023PhD} 
%After the definition of the constraint tensor in \cite{Mars2023first}, 
the constraint tensor was defined for \textit{general} hypersurface data (i.e.\ non-necessarily null) and in terms of the metric hypersurface connection $\nablao$ (instead of $\ovnabla$), so that 
%.\ In this way, the tensor 
$\bY$ appears explicitly.\  
%in the definition.\ 
In the null case, this tensor has
already been exploited in the recent works \cite{mars2024transverseI, mars2024transverseII}.

Let us construct the constraint tensor 
in the null case with fully explicit dependence of the tensor $\bY$.\ 
%already mentioned notion of constraint tensor in which $\bY$ appears explicitly.\ 
For that, we 
%The notion of constraint tensor in the null case arises naturally from the fact that 
% We 
show that 
%a certain linear combination of 
the tangential components of the ambient Ricci tensor 
%and of the transverse-tangential-transverse-tangential components of the ambient Riemann tensor 
can be computed exclusively in terms of the hypersurface data (in the embedded case).\ This leads naturally to the definition, on any null hypersurface data, of a symmetric $(0,2)$-tensor that encodes the tangent part of the ambient Ricci tensor 
at the \textit{purely 
	abstract level}.\

Consider null hypersurface data $\hypdata$ $\{\phi,\rig\}$-embedded in  $(\mathcal{M}, g)$ and assume Setup \ref{setup:basis:e_a}.\ 
%We intend to compute the contraction 
%ambient Ricci tensor along tangential directions to $\phi (\mathcal{N})$, i.e.\ 
%$R_{\alpha\beta}e_b^{\alpha}e_d^{\beta}$ at $\phi (\mathcal{N})$.\ 
Combining 
%From 
$R_{\alpha\beta}\defi  g^{\mu\nu}R_{\mu\alpha\nu\beta}$ with 
%and inserting 
the first expression in \eqref{gup},   
it follows
\begin{align*}
	R_{\alpha\beta}e_b^{\alpha}e_d^{\beta}& \stackbin{\phi (\mathcal{N})}{=} 
	%
	%g^{\mu\nu}R_{\mu\alpha\nu\beta}e_b^{\alpha}e_d^{\beta}
	%
	%\stackbin{\phi (\mathcal{N})}{=} 
	%
	\Big( n^c\lp\rig^{\mu} e_c^{\nu}+\rig^{\nu}e_c^{\mu}\rp+P^{ac}e_a^{\mu} e_c^{\nu}\Big) R_{\mu\alpha\nu\beta}e_b^{\alpha}e_d^{\beta}\\
	& \stackbin{\phi (\mathcal{N})}{=}  
	n^c\lp R_{\mu\alpha\nu\beta}\rig^{\mu}e_b^{\alpha}e_c^{\nu}e_d^{\beta} + R_{\nu\beta\mu\alpha}\rig^{\nu}e_d^{\beta}e_c^{\mu} e_b^{\alpha}\rp + P^{ac} R_{\mu\alpha\nu\beta}e_a^{\mu}e_b^{\alpha}e_c^{\nu}e_d^{\beta},
\end{align*}
%We thus write 
%The previous identity 
which 
can be rewritten 
in terms of the Ricci tensor $\textbf{Ric}$ and the Riemann tensor $\textbf{Riem}$ of $g$ as
%(recall \eqref{gnunu})
\begin{align}
	\textbf{Ric} (e_b, e_d) \stackbin{\phi (\mathcal{N})}{=}  2n^c \textbf{Riem} (\rig,e_{(b\vert},e_c,e_{\vert d)})+ P^{ac} \textbf{Riem} (e_a,e_b,e_c,e_d).\label{ricciabstract1} 
\end{align}
%where $\textbf{Ric}$ and $\textbf{Riem}$ are respectively the Ricci and Riemann tensors of $(\M,g)$.\ 
%
By \eqref{Codazzi}-\eqref{Gauss} 
%Proposition \ref{Riemnablao} 
we know that  $R_{\mu\alpha\nu\beta}\rig^{\mu}e_b^{\alpha}e_c^{\nu}e_d^{\beta}$ and  $R_{\mu\alpha\nu\beta}e_a^{\mu}e_b^{\alpha}e_c^{\nu}e_d^{\beta}$ can be fully expressed in terms of the hypersurface data.\ 
%\tcr{However, in general this is not true for the components $R_{\mu\alpha\nu\beta}\rig^{\mu}e_b^{\alpha}\rig^{\nu}e_d^{\beta}$.} 
This, together with \eqref{ricciabstract1}, %Consequently, 
%This 
means 
%, in particular, 
that 
the pull-back $\phi^{\star}\textbf{Ric}$ 
%to $\N$ of the $\textbf{Ric}$ 
%ambient Ricci tensor 
can be fully codified by $\{\gamma,\ellc,\elltwo,\bY\}$.\  
It therefore 
makes sense to find an explicit expression for the right-hand side of \eqref{ricciabstract1}.\  
%in terms of $\hypdata$. 
For that we no longer need  to assume that the data is embedded.\ Instead, we introduce two tensors 
$\B_{bcd}$ and $\C_{abcd}$ on $\N$, which correspond to the
hypersurface data counterparts of $R_{\mu\alpha\nu\beta}\rig^{\mu}e_b^{\alpha}e_c^{\nu}e_d^{\beta}$ and
$R_{\mu\alpha\nu\beta}e_a^{\mu}e_b^{\alpha}e_c^{\nu}e_d^{\beta}$ respectively (as given in \eqref{Codazzi}-\eqref{Gauss}).\
%Proposition \ref{Riemnablao}). 
The right-hand side
of \eqref{ricciabstract1} can then be elaborated at the abstract level by computing the contractions $n^c(\B_{bcd}+ \B_{dcd})$ and $P^{ac} \C_{abcd}$.\ Let us put forward the definition of 
$\B_{bcd}$ and $\C_{abcd}$ as dictated by \eqref{Codazzi}-\eqref{Gauss}. 
%Proposition \ref{Riemnablao}.
\begin{definition}%(Tensors $\B$ and $\C$)
	\label{def:A:and:B:tensors}
	Given null hypersurface data $\hypdata$, the tensors $\B$, $\C$ are defined as
	\begin{align}
		%\nn 
		\B_{bcd} \defi &  \spc 
		\ell_a \Riemo{}^a_{\,\,\,bcd} + 2 \nablao_{[d} {\Y}_{c]b} 
		+ 2 \elltwo \nablao_{[d} \U_{c]b}  +  \U_{b[c} \nablao_{d]} \elltwo  
%		\\
%		& \spc 
		+ {\Y}_{b[d} \left ( 2 \left ( \F_{c]f} + {\Y}_{c]f} \right ) n^f  \right ) ,
		\label{defB} \\
		\C_{abcd} \defi & 
		\gamma_{af} \Riemo{}^f_{\,\,\,bcd}
		+ 2 \ell_a \nablao_{[d} \U_{c]b} 
		+ 2 {\Y}_{b[c} \U_{d]a} + 2 \U_{b[c}  \left ( {\Y}_{d]a} + \F_{d]a} \right ). \label{defD}
	\end{align}
\end{definition}
Our guiding principle to compute $n^c(\B_{bcd}+ \B_{dcd})$ and $P^{ac} \C_{abcd}$ 
is to write down as many derivatives of $\bY$ as possible in terms of $\pounds_{n} \bY$, i.e.\ as evolution terms along the direction $n$.\ In many circumstances, this constitutes a great advantage because it allows one to 
%This is of great use because it allows one to 
understand the identities as transport equations for the data tensor $\bY$.\ In fact, in Section \ref{sec:application} we shall prove that given $\bY$ on a cross-section of the data, one can determine the full tensor $\bY$ everywhere on $\N$ by integrating the identities below. 
%turnbYs out to be particularly useful  
%in the null case, where $n$ is the degeneration direction of $\gamma$.\
%The result, however, holds in full generality.
\begin{proposition}\label{propnBPD}
	Let $\hypdata$ be null hypersurface data and $\bs{r}$, $\Q$ be given by \eqref{defY(n,.)andQ}. Then, the tensors $\B$ and $\C$ introduced in Definition \ref{def:A:and:B:tensors} satisfy the following identities:
	\begin{align}
		P^{ac} \C_{abcd} = & \spc\Riemo_{(bd)}
		- \nablao_{(b} \sone_{d)}
		+ \sone_{b} \sone_{d}
		- n(\elltwo)  \U_{bd} 
		%\nonumber \\
		%&
		 +  2P^{ac}
		\Big(\U_{b[a} \Y_{d]c}
		+ \U_{a[d} \Y_{c]b} \Big), \label{traceD}\\
		\nonumber  n^c \left ( \B_{bcd} + \B_{dcb} \right ) = &
		- 2 \pounds_{{n}} \Y_{bd}
		+   2 \nablao_{(b} \left ( \sone_{d)} + \Yn_{d)} \right )
		- 2\Q  \Y_{bd}- 2 \left ( \Yn_b - \sone_b \right ) \left ( \Yn_d - \sone_d \right )\\
		&+  n( \elltwo) \U_{bd}    
		+ 2 P^{ac} \left ( \Y_{c(b} - \F_{c(b} \right ) 
		\U_{d)a}  .
		\label{nBsymLemma}
	\end{align}
	%\begin{equation}
	%  \label{defY(n,.)andQ}
	%  \bs{\Yn}\defi \bY(n,\cdot)\qquad\text{and}\qquad\Q \defi \bY(n,n). 
	%\end{equation}     
	%Defining $\trP\bU\defi P^{ab}\U_{ab}$, $\trP\bY\defi P^{ab}\Y_{ab}$, 
	Moreover, it also holds
	\begin{align}
		\nonumber  n^c &\left ( \B_{bcd} + \B_{dcb} \right) +P^{ac} \C_{abcd} =  \Riemo_{(bd)}- 2 \pounds_{{n}} \Y_{bd}- \lp 2\Q +\trP\bU \rp\Y_{bd}\\
		&\hspace{0.4cm}+ \nablao_{(b} \lp \sone_{d)}+ 2  \Yn_{d)}\rp -2\Yn_b\Yn_d+4\Yn_{(b}\sone_{d)}-s_bs_d-(\trP\bY)\U_{bd}+ 2P^{ac}\U_{a(b}\lp 2Y_{d)c}+\F_{d)c}\rp.\label{sumnBandPC}
	\end{align}
\end{proposition}
\begin{proof}
	We start by computing $\ell_f \Riemoin{}^f{}_{bcd} n^c$. 
	%
	%  %
	%
	%  \tcb{In Appendix \ref{app:curvature:null}, we have provided an expression for 
		%    (cf.\ \eqref{Riemoellnnull}).} However, for this proof it is useful to use an alternative form.
	The Ricci identity applied to $\ell_b$ gives
	\begin{align}
		\ell_f \Riemoin{}^f{}_{bcd}
		& = \nablao_d \nablao_c \ell_b - \nablao_{c} 
		\nablao_d \ell_b \stackbin{\eqref{nablaoll}}{=}
		\nablao_d \left ( \F_{cb} - \elltwo \U_{cb} \right ) -
		\nablao_c \left ( \F_{db} - \elltwo \U_{db} \right ).
		\label{ellRiem}
	\end{align}  
	Contracting this with $n^c$ and using \eqref{Codazzitos2}
	applied to $A \rightarrow \bF$ and $\mathfrak{a} \rightarrow \bs{\sone}$ one gets, after inserting $d\bF =0$ (which follows from
	the definition $\bF = \frac{1}{2} d \ellc$),
	\begin{align}
		\ell_f \Riemoin{}^f{}_{bcd} n^c & =
		\nablao_b \sone_d - \F_{cd} \nablao_b n^c - 
		2 \elltwo n^c \nablao_{[d} \U_{c]b}
		+ \U_{bd} n( \elltwo)  - n^c \U_{cb} \nablao_d \elltwo.
		\label{ellRiemn}
	\end{align}
	By  \eqref{antisymRiemo} the Ricci tensor  $\Riemo_{ab}$ can be written as
	\begin{align}
		\Riemo_{bd} & = \Riemo_{(bd)} + \Riemo_{[bd]} 
		= \Riemo_{(bd)}
		+ \nablao_{[b} \sone_{d]} . \label{whatever}
	\end{align}
	Combining \eqref{prod3}-\eqref{prod4} with \eqref{ellRiemn} we then obtain 
	%   Recalling \eqref{prod3}-\eqref{prod4} we then obtain, from this and
	%\eqref{ellRiemn}
	\begin{align}
		 P^{ac}  \left ( \gamma_{af} \Riemo{}^f{}_{bcd}  + 2 \ell_a \nablao_{[d} \U_{c]b} \right )
		&=  \Riemo_{bd} - n^c \ell_f \Riemo{}^f{}_{bcd} - 2 \elltwo n^c
		\nablao_{[d} \U_{c]b}  \nonumber  \\
		& =  \spc \Riemo_{(bd)} - \nablao_{(b } \sone_{d)}
		- \U_{bd} n ( \elltwo)   + \F_{cd} \nablao_b n^c
		+ n^c \U_{cb} \nablao_d \elltwo.  \label{term1}
	\end{align}
	We elaborate the last two terms by taking into account \eqref{Un} and \eqref{nablaonnull}. This yields
	\begin{align}
		\F_{cd} \nablao_b n^c
		+ n^c \U_{cb} \nablao_d \elltwo
		= & 
		- P^{ac}   \U_{ba} \F_{dc}  + \sone_b\sone_d. \label{term2}
	\end{align}
	We have all the ingredients to compute $P^{ac} \C_{abcd}$.\ Contracting the right hand side of
	\eqref{defD} with $P^{ac}$ and replacing \eqref{term1} and \eqref{term2}, expression \eqref{traceD} follows after simple manipulations.

	For \eqref{nBsymLemma} we start by substituting
	\eqref{ellRiem} in \eqref{defB}, which  gives
	\begin{align}
		\label{ExpB}
		\hspace{-0.4cm}\B_{bcd} = &\spc 
		\nablao_d \F_{cb}  - \nablao_c  \F_{db} 
		+ \nablao_d \Y_{cb} 
		- \nablao_c \Y_{db}
		- \frac{1}{2}  \U_{cb}   \nablao_d \elltwo
		\nonumber \\
		&+ \frac{1}{2}  \U_{db}   \nablao_c \elltwo + \Y_{bd} \left ( \F_{cf} + \Y_{cf} \right ) n^f 
		- \Y_{bc} \left ( \F_{df} + \Y_{df} \right ) n^f .
	\end{align}
	We now contract with $n^c$ and use \eqref{Un}, \eqref{CodazziToLie} with $S \rightarrow \bY$ and \eqref{Codazzitos} with $A \rightarrow \bF$ to get %\tcr{(recall that  added because it was necessary)}
	\begin{align}
		n^c \B_{bcd} = &\spc 
		\nablao_{(b} \sone_{d)}
		- \F_{c(b} \nablao_{d)} n^c
		-\frac{1}{2} n^c \nablao_c \F_{db}
		+ 
		\nablao_d  \Yn_{b} 
		- \pounds_{{n}} \Y_{bd}  + \Y_{cd} \nablao_b n^c
		\nonumber \\
		&  + \frac{1}{2}  n( \elltwo)\U_{db}   - \Q  \Y_{bd}
		+\Yn_{b} \sone_d
		- \Yn_{b} \Yn_{d},  \label{nB} 
	\end{align}
	where we have taken into account the definitions \eqref{defY(n,.)andQ}.\  
	%also replaced $\Y_{bc} n^c$ and
	%$\Y_{bc} n^b n^c$ by $\Yn_b$ and $\Q $ respectively.
	The symmetric part is therefore given by %Taking the symmetric part one obtains
	\begin{align}
		\nn n^c \left ( \B_{bcd} + \B_{dcb} \right ) = &\spc 
		2 \nablao_{(b} \left ( \sone_{d)} + \Yn_{b)} \right )
		+ 2 \left ( \Y_{c(b} - \F_{c(b} \right ) \nablao_{d)} n^c
		- 2 \pounds_{{n}} \Y_{bd}\\
		&  - 2\Q  \Y_{bd}+ 2 \Yn_{(b} \sone_{d)}- 2 \Yn_{b} \Yn_{d} 
		+  n( \elltwo) \U_{bd}  .  \label{nBsym} 
	\end{align}
	By virtue of \eqref{nablaonnull}, we finally find
	\begin{align*}
		2 (\Y_{cb} -\F_{cb}) \nablao_d n^c =&\spc 
		2 P^{ac} \U_{da}  \left ( \Y_{cb} - \F_{cb} \right ) 
		+ 2 \sone_d\left ( \Yn_b - \sone_b \right )  ,
	\end{align*}
	which together with \eqref{Un} yields \eqref{nBsymLemma} when inserted into \eqref{nBsym}. Finally, equation \eqref{sumnBandPC} follows at once after simple index manipulations.
\end{proof}
Note that the right hand side of \eqref{traceD} is explicitly symmetric in the indices $b,d$. This property is consistent with the fact that, in the embedded case, the left-hand side of \eqref{Gauss} is symmetric under the interchange of the first and second pair of indices.\ This provides a non-trivial consistency check for  \eqref{traceD}.

%As explained above, 
Expression \eqref{sumnBandPC} motivates introducing a symmetric tensor 
$\Rtensor$ 
on $\N$ that one calls \textbf{constraint tensor} (cf.\ \cite{manzano2023PhD}). 
\begin{definition}
	\label{Rtensor_null}
	(Constraint tensor $\Rtensor$)
	Given null hypersurface data $\hypdata$, the constraint tensor $\Rtensor$ tensor is the symmetric $(0,2)$-tensor
%	\begin{align}
%		\nonumber \Rtensor_{bd} \defi  & \Riemo_{(bd)}- 2 \pounds_{{n}} \Y_{bd}
%		- \lp 2\Q +\trP\bU \rp\Y_{bd}+ \nablao_{(b} \lp \sone_{d)}+ 2  \Yn_{d)}\rp -2\Yn_b\Yn_d+4\Yn_{(b}\sone_{d)}-s_bs_d\\
%		&-(\trP\bY)\U_{bd}+ 2P^{ac}\U_{a(b}\lp 2\Y_{d)c}+\F_{d)c}\rp.\label{defabsRicci}
%	\end{align}
	\begin{align}
		\nn \Rtensor_{ab} \defi  &\spc \Riemo_{(ab)}- 2 \pounds_{{n}} \Y_{ab}
		- \lp 2\Q+\trP\bU \rp\Y_{ab}+ \nablao_{(a} \lp \sone_{b)}+ 2  \Yn_{b)}\rp  \\
		\label{defabsRicci} & -2\Yn_a\Yn_b+4\Yn_{(a}\sone_{b)}-s_as_b-(\trP\bY)\U_{ab}+ 2P^{cd}\U_{d(a}\lp 2\Y_{b)c}+\F_{b)c}\rp,
	\end{align}
	where $\Q$ and $\Yn_a$ are defined by \eqref{defY(n,.)andQ}. 
\end{definition}
The whole construction has been performed so that the following result holds.
\begin{proposition}
	\label{RtensorHyp}
	Let $\{ \mathcal{N},\gamma,\ellc,\elltwo,\bY \}$ be null hypersurface data $\{\phi,\rig\}$-embedded in $(\M,g)$.\ 
	%\tcr{Let $\bs{\nu}$ be the unique normal covector along $\phi (\mathcal{N})$ satisfying $\bs{\nu}(\zeta) = 1$ and define $\nu\defi g(\bs{\nu},\cdot)$.} 
%	Consider the symmetric $(0,2)$-tensor
%	\begin{align*}
%		\bm{\Rtensor}\defi \textup{\textbf{Ric}} 
%	\end{align*}
%	along $\phi(\N)$.\ 
	Then the Ricci tensor $\textup{\textbf{Ric}} $ of $g$ satisfies
	\begin{equation}
	\label{RicciIsPullbackRicci_0} 
	\phi ^{\star}  \textup{\textbf{Ric}}  = \Rtensor.
	\end{equation}
%	\begin{equation}
%		\label{RicciIsPullbackRicci_0} 
%		\phi ^{\star}  \bm{\Rtensor}  = \Rtensor.
%	\end{equation}
%	In particular at any point $p$ where the hypersurface $\phi(\N)$ is null, it holds
%	\begin{equation}
%		\label{RicciIsPullbackRicci} 
%		\phi ^{\star} \textup{\textbf{Ric}} |_p = \Rtensor |_p.
%	\end{equation}
\end{proposition}
%In the null case $\ntwo=0$, 
%The 
Conditions $\Rtensor=\Lambda\gamma$, $\Lambda\in\mathbb{R}$ can be thought of as 
%the 
vacuum constraint equations with cosmological constant $\Lambda$ on a null hypersurface.\  Such constraints have always appeared in the literature in a decomposed form adapted to a foliation by spacelike slices.\ To the best of our knowledge, the only exceptions to this 
are the already mentioned works  \cite{manzano2023PhD, mars2023covariant,Mars2023first}.\ 
%,  where the abstract nature of equations $\Rtensor=\Lambda\gamma$ was exploited in the context of the characteristic initial value problem.\ 

As anticipated before, Definition \ref{Rtensor_null} shows the full dependence on $\bY$, 
%(in the terms involving $\bY$, $\bs{\Yn}$ and $\Q$),
as both $\nablao$ and $\Ricco$ depend only on the metric part of the data.\ 
Furthermore, the result \eqref{defabsRicci} involves no decomposition w.r.t.\ any foliation.\ In fact, it makes no assumption on whether such foliation exists.\ Definition \ref{Rtensor_null} is fully covariant on $\N$, even though this manifold admits no metric.\ We emphasize that it is by use of the hypersurface data formalism (in particular thanks to the existence of the connection $\nablao$) that such compact, unified form 
of the vacuum constraint equations in the null case becomes possible.

Observe that, from its interpretation in the embedded case, it is to be expected that the constraint tensor is gauge invariant at a null point.\ This was proven in \cite[Theorem 4.6]{Mars2023first} in the case of characteristic hypersurface data, namely null hypersurface data with product topology.\ However, the proof does not make use of this topological restriction, so gauge invariance of $\Rtensor$  
%that can be foliated by diffeomorphic sections with positive definite induced metric. However, the proof of Theorem 4.6 in \cite{Mars2023first} does not rely on  these global restrictions, so the gauge invariance of the constraint tensor 
holds for general null hypersurface data.\ 
%\footnote{It should actually be true that gauge invariance holds even at isolated null points. We do not attempt proving this fact here.}. 
In particular, this means that we can compute $\Rtensor$ in any gauge, which gives a lot of flexibility to adjust the gauge to the problem at hand. %At non-null points gauge invariance does not hold since the spacetime tensor $\bm{\Rtensor}$ depends on the rigging vector $\rig$. 

%In particular, since 
%the vector field $n$ defines a privileged direction on a null hypersurface data (because the action of the gauge group on $n$ leaves such direction invariant) 
We devote the following two sections to analyse some geometrical properties of the constraint tensor.\ Specifically, we 
%aim to 
find its components along the degenerate direction given by $n$, as well as its pull-back to any 
%relationship with the geometry of any 
transverse 
%codimension-one 
submanifold $S$ 
%(non-necessarily a cross-section) 
of $\N$.\ Both results combined provide information on all components of the constraint tensor.\ 
%In Section \ref{sec:Constraint(n,-)} we compute the contractions with $n$ and in Section \ref{secRiemoAB} the pull-back of $\Rtensor$ to the submanifold $S$.

\subsection{Constraint tensor along the null direction $n$}\label{sec:Constraint(n,-)}

The following theorem finds the explicit expressions for the contractions $\Rtensor(n,\cdot)$, $\Rtensor(n,n)$.\ For that we will need the contraction $\Riemo_{(ab)}n^a$, which has been computed in Appendix \ref{app:curvature:null} (see Lemma \ref{lem:Riemo:contracted:with:n}).\ We emphasize that the result does not require any topological assumption on
$\N$. In particular, the null hypersurface data does not need to be foliated by sections.
\begin{theorem}\label{thmR(n,-)andR(n,n)}
	Consider null hypersurface data $\hypdata$.\ The constraint tensor $\Rtensor$ verifies
	%Then, 
	\begin{align}
		\label{ConstTensror(n,-)} \hspace{-0.25cm} \Rtensor_{ab}n^a &=  -\nablao_{b}\Q \hspace{-0.03cm} - \hspace{-0.03cm}\pounds_n(\Yn_b-\sone_{b}) \hspace{-0.03cm} -\hspace{-0.03cm} (\trP\bU)\lp \Yn_{b} \hspace{-0.03cm} -\hspace{-0.03cm} \sone_b\rp \hspace{-0.03cm} -\hspace{-0.03cm} \nablao_b(\textup{tr}_P\bU)\hspace{-0.03cm} +\hspace{-0.03cm} P^{cd} \left ( \nablao_c\U_{bd}\hspace{-0.03cm} -\hspace{-0.03cm} 2 \U_{bd}\sone_{c} \right ) ,\\
		\label{ConstTensror(n,n)}\hspace{-0.25cm} \Rtensor_{ab}n^an^b &= -n(\textup{tr}_P\bU)+(\trP\bU) \Q -P^{ab}P^{cd}\U_{ac}\U_{bd}.
	\end{align}
\end{theorem}
\begin{proof}
	Recall the results $\bU(n,\cdot)=0$, $\bs{\sone}(n)=0$, $\bY(n,\cdot)=\bs{\Yn}$, $\bY(n,n)=\bs{\Yn}(n)=-\Q$, and  $\bF(n,\cdot)=\bsone$.\ Particularizing \eqref{contrNsym} 
	for %$\ntwo=0$ and 
	$\bs{\theta}=\bsone+2\bs{r}$ we get
	\begin{align}
		\label{mideq32}n^a \nablao_{(a}(\sone_{b)}+2\Yn_{b)})=&\spc\dfrac{1}{2}\pounds_{n}\sone_{b}+\pounds_{n}\Yn_{b}- \nablao_{b}\Q+2\Q  \sone_b - P^{ac} \U_{bc}(\sone_{a}+2\Yn_{a}).
	\end{align}
	The contraction of \eqref{defabsRicci}
	%\eqref{defabsRiccinull} 
	with $n^a$ 
	gives \eqref{ConstTensror(n,-)} after  inserting \eqref{Riemosym(n,-)}, \eqref{mideq32}
	and 
	$\pounds_n\Yn_b=n^a\pounds_{n}\Y_{ab}$. Contracting \eqref{ConstTensror(n,-)}  with $n^b$ and using that $n^bP^{cd}\nablao_c\U_{bd}=-P^{cd}\U_{bd}\nablao_cn^b=-P^{ab}P^{cd}\U_{ac}\U_{bd}$ as well as $n^b\pounds_n(\sone_b-\Yn_b)=n(\Q)$ yields \eqref{ConstTensror(n,n)}.
\end{proof}
Observe that the identity \eqref{ConstTensror(n,n)} corresponds to the null Raychaudhuri equation \cite{carter2012gravitation, gourgoulhon20063+,wald1984general}.\ Indeed, from the comparison between \eqref{ConstTensror(n,n)} and \cite[Eq.\ (6.8)]{gourgoulhon20063+}, 
%\eqref{Raychaudhuri:original}, 
it is straightforward to conclude that $\text{tr}_P\bU$ plays the role of the expansion $\theta$ at the abstract level, while $P^{ab}P^{cd}\U_{ac}\U_{bd}$ stands for the term $(\n-1)^{-1}\theta^2+\sigma_{ab}\sigma^{ab}$, where $\sigma$ is the shear tensor.
%expressed in terms of null hypersurface data. %\tcr{[Maybe anticipate here some later result on horizons (where $\bU=0$)]}

\subsection{Constraint tensor on a transverse submanifold $S$}\label{secRiemoAB}

Let us now assume Setup \ref{setup} and analyze the case when we have selected a codimension-one submanifold $S$ of $\N$ to which $n$ is everywhere transverse.\ In particular, all results from Section \ref{seclevicivitaonS} can be
applied.\ Our main aim is to derive an explicit expression for the pull-back to $S$ of the constraint tensor, i.e.\ $\psi^{\star}\mathcal{R}$, in terms of the Ricci tensor of the Levi-Civita connection $\nabla^h$ (see Section \ref{seclevicivitaonS}).\ In the context of the characteristic problem in General Relativity studied in \cite{mars2023covariant, Mars2023first},  
the constraint tensor evaluated on non-degenerate submanifolds has played a relevant role.\ The analysis there, as already mentioned, is different to the present one
in several respects.\ Firstly, instead of Lie derivatives along $n$, a different evolution operator was used.\ Secondly, and more importantly, the expressions we find here are fully explicit in the tensor $\bY$.\ 
%This is very relevant in many situations, e.g. to integrate the vacuum constraint equations or to study spacetime matchings and, in particular, the thin shells they produce.\ 
Last, but no least, our identities 
%the expressions here 
are completely general, while
in \cite{mars2023covariant, Mars2023first} a specific gauge was chosen from the outset.   
%  In the former, it was decomposed in terms of a foliation by spacelike leaves while in the latter its pull-back to a non-degenerate submanifold was obtained under specific gauge conditions.}

By \eqref{defabsRicci}, 
%\eqref{defabsRiccinull}, 
the task of computing $\psi^{\star}\mathcal{R}$
%
%this task 
%
requires relating the pull-back $\psi^{\star}\Ricco$ with the Ricci tensor of $\nabla^h$. Now, computing the pull-back $\psi^{\star}\Ricco$ amounts to calculating $\Riemo_{AB}\defi  \Riemo{}^c{}_{acb} v^a_Av^b_B$, and this trace can be obtained by means of \eqref{prod4} and \eqref{Pdecom:abstract} as follows:
\begin{align}
\nn \Riemo_{AB}=&\spc \delta^{c}_f\Riemo{}^f{}_{acb} v^a_Av^b_B= \lp P^{cd}\gamma_{df}+n^c\ell_f\rp \Riemo{}^f{}_{acb} v^a_Av^b_B\\
\label{RiemoABalmosttheend}=&\spc \lp h^{CD}v_C^cv_D^d\gamma_{df}+n^c(\ell_f-h^{CD}\ellp_Cv_D^d\gamma_{df})\rp\Riemo{}^f{}_{acb} v^a_Av^b_B.
\end{align}
Thus, we need to evaluate both 
\begin{align*}
h^{CD}v_D^d\gamma_{df}\Riemo{}^f{}_{acb} v^a_Av_C^cv^b_B\qquad\text{and}\qquad n^c(\ell_f-h^{CD}\ellp_Cv_D^d\gamma_{df})\Riemo{}^f{}_{acb} v^a_Av^b_B.
\end{align*}
The first one is obtained by contracting \eqref{gammaRiemanvvvOURCASE} with $h^{CD}$.
For the second one, substituting \eqref{PUandPUU}
%-\eqref{trPYandtrPUwithPdec} %\eqref{Pdecomp2} 
into \eqref{Riemoellnnull} and \eqref{gammaRiemonthird2} yields 
\begin{align}
\nn n^c\Big(\ell_f- h^{CD}\ellp_Cv_D^d\gamma_{df}\Big) &\Riemo{}^f{}_{acb}v_A^av_B^b= v_A^av_B^b\nablao_a \sone_b- \sone_A\sone_B+\lp n(\elltwo) -2\ellp^{D}\sone_D\rp\U_{AB}\\
\nn &+( \elltwo -\elltwo_{\parallel}) (\pounds_{n}\bU)_{AB}+2\ellp^{D}\lp  \sone_{(A}\U_{B)D}-v_D^dv_A^av_B^b\nablao_{[a}\U_{d]b}\rp\\
\label{mideq27} &-h^{CD}\U_{AC}\lp \F_{DB}+(\elltwo -\elltwo_{\parallel}) \U_{BD}\rp.
\end{align}
We elaborate \eqref{mideq27} by particularizing \eqref{covderpcovtensoronS} for $\mathcal{T}=\sone$, $\mathcal{T}=\bU$ and $\mathcal{T}=\ellc$. Since $s(n)= \bU(n,\cdot)=0$ they give, respectively,
\begin{align}
\label{nablasoneAB}v_A^av_B^b\nablao_a \sone_b=&\spc\nabla^h_A\sone_B-\ell^C\sone_C\U_{AB},\\
2v_D^dv_A^av_B^b\nablao_{[a}\U_{d]b} %=&\spc v_D^dv_A^av_B^b\lp \nablao_{a}\U_{db}-\nablao_{d}\U_{ab}\rp\\
\label{nablaUpAB}=&\spc\nabla^h_A\U_{BD}-\nabla^h_D\U_{AB}-\ell^C\U_{CD}\U_{AB}+\ell^C\U_{AC}\U_{BD},\\
\label{FAB}v_A^av_B^b \nablao_{a}\ell_{b}=&\spc \nabla^h_{A}\ellp_{B}-\elltwo_{\parallel}\U_{AB}\quad\Longrightarrow\quad\F_{AB}%\defi v_A^av_B^b\F_{ab}
=v_A^av_B^b \nablao_{[a}\ell_{b]}=\nabla^h_{[A}\ellp_{B]},
\end{align}
with which \eqref{mideq27} becomes 
\begin{align}
\nn n^c \Big(\ell_f-h^{CD}\ellp_C & v_D^d\gamma_{df}\big)  \Riemo{}^f{}_{acb}v_A^av_B^b= \nabla_{A}^h\sone_{B}- \sone_A\sone_B+\lp n(\elltwo) +\ellp^{C}\ellp^D\U_{CD} -3\ellp^{C}\sone_C\rp \U_{AB}\\
\nn &+(\elltwo -\elltwo_{\parallel})(\pounds_{n}\bU)_{AB}+ 2\ellp^{C}\sone_{(A}\U_{B)C}+\ellp^{C}\nabla_{C}^h\U_{A B}-\ellp^{C}\nabla_{A}^h\U_{C B}\\
\label{2ndTERM} &-\lp h^{CD}(\elltwo -\elltwo_{\parallel})+\ellp^C\ellp^{D}\rp\U_{AC} \U_{BD} -\frac{1}{2}h^{CD}(\nabla^h_{D}\ellp_{B}-\nabla^h_{B}\ellp_{D})\U_{AC}.
\end{align}
The Ricci tensor $\Riemo_{AB}$ follows by substituting \eqref{2ndTERM} 
%\eqref{1stTERM} 
and \eqref{gammaRiemanvvvOURCASE} (contracted with $h^{CD}$) into \eqref{RiemoABalmosttheend}:
\begin{align}
\nn \Riemo_{AB}=&\spc R^h_{AB}+\nabla_{A}^h\sone_{B}- \sone_A\sone_B+\Big( n(\elltwo) +2\ellp^C\ellp^D\U_{CD} -3\ellp^{C}\sone_C+ \nabla^h_C\ellp^C\\
\nn &+  (\textup{tr}_h\bUp)(\elltwo-\elltwo_{\parallel})\Big) \U_{AB}+(\elltwo -\elltwo_{\parallel})(\pounds_{n}\bU)_{AB}+(\textup{tr}_h\bUp)\nabla^h_{(A}\ellp_{B)}\\
\nn &+2\ellp^{C}\lp \nabla_{C}^h\U_{A B}+ \sone_{(A}\U_{B)C}-\nabla_{(A}^h\U_{B)C}\rp-2\Big( h^{CD}(\elltwo -\elltwo_{\parallel})\\
\label{RiemoAB_0} &+\ellp^C\ellp^{D}\Big)\U_{AC} \U_{BD}-h^{CD}\lp \U_{DB}\nabla^h_{(A}\ellp_{C)}+ \U_{DA}\nabla^h_{(B}\ellp_{C)}\rp.
\end{align}
Note that all terms in \eqref{RiemoAB_0} except from $\nabla_{A}^h\sone_{B}$ are symmetric.\  
This implies that $\Riemo_{[AB]}=\nabla_{[A}^h\sone_{B]}$, which is in agreement with equation \eqref{antisymRiemo} and provides a consistency check to \eqref{RiemoAB_0}.\ The symmetrized tensor is
\begin{align}
\nn \Riemo_{(AB)}=&\spc R^h_{AB}+\nabla_{(A}^h\sone_{B)}- \sone_A\sone_B+\Big( n(\elltwo) +2\ellp^C\ellp^D\U_{CD} -3\ellp^{C}\sone_C+ \nabla^h_C\ellp^C\\
\nn &+  (\textup{tr}_h\bUp)(\elltwo-\elltwo_{\parallel})\Big) \U_{AB}+(\elltwo -\elltwo_{\parallel})(\pounds_{n}\bU)_{AB}+(\textup{tr}_h\bUp)\nabla^h_{(A}\ellp_{B)}\\
\nn &+2\ellp^{C}\lp \nabla_{C}^h\U_{A B}+ \sone_{(A}\U_{B)C}-\nabla_{(A}^h\U_{B)C}\rp-2\Big( h^{CD}(\elltwo -\elltwo_{\parallel})\\
\label{RiemoAB_1} &+\ellp^C\ellp^{D}\Big)\U_{AC} \U_{BD}-h^{CD}\lp \U_{DB}\nabla^h_{(A}\ellp_{C)}+ \U_{DA}\nabla^h_{(B}\ellp_{C)}\rp.
\end{align}
Having obtained \eqref{RiemoAB_1}, we can now write down the relation between the pull-back to $S$ of the constraint tensor and the Ricci tensor of the induced metric $h$.
\begin{theorem}\label{thmRconstraintAB}
Consider null hypersurface data $\hypdata$ and assume Setup \ref{setup}.\ Let $R^{h}_{AB}$ be the Ricci tensor of the Levi-Civita connection $\nabla^h$ on $S$. Then, the pull-back %$\psi^{\star}\Rtensor$ 
to $S$ of the constraint tensor $\Rtensor$ defined by \eqref{defabsRicci} is given by
\begin{align}
	\nn \Rtensor_{AB} =&\spc  R^h_{AB} +2\nabla^h_{(A} \lp \sone_{B)}+r_{B)}\rp-2(\Yn_A-\sone_A)(\Yn_B-\sone_B)\\
	\nn &+(\elltwo -\elltwo_{\parallel})(\pounds_{n}\bU)_{AB}- 2 (\pounds_{{n}} \bY)_{AB}- \lp 2\Q +\textup{tr}_h\bUp \rp\lp \Y_{AB}-\nabla^h_{(A}\ellp_{B)}\rp\\
	\nn &+\lp n(\elltwo) +2\ellp^C\ellp^D\U_{CD} -4\ellp^{C}\sone_C + \lp \textup{tr}_h\bUp+\Q\rp(\elltwo-\elltwo_{\parallel}) -\textup{tr}_h\bYp+ \nabla^h_C\ellp^C \rp \U_{AB}\\
	\nn &+2\ellp^{C}\lp \nabla_{C}^h\U_{A B}-\nabla_{(A}^h\U_{B)C}-\lp 2\lp \Yn_{(A}-\sone_{(A}\rp+\ellp^{D}\U_{D(A} \rp \U_{B)C}\rp\\
	\label{RconstraintABfinal}
	&+ 2h^{CD}\lp 2\Y_{D (A}-\nabla^h_{D}\ellp_{ (A}-(\elltwo -\elltwo_{\parallel})\U_{D(A} \rp \U_{B)C}.
\end{align}
\end{theorem}
\begin{proof}
We need to multiply \eqref{defabsRicci} 
%\eqref{defabsRiccinull} 
by $v_A^av_B^b$.\ One comes across a term $v_A^av_{B}^{b}\nablao_{(a}( \sone_{b)} +2r_{b)})$ which we elaborate by using \eqref{covderpcovtensoronS} for $\mathcal{T}=\bsone$ and $\mathcal{T}=\bs{r}$ (recall that $\bsone(n)=0$, $\bs{r}(n)=-\Q$), thus obtaining
\begin{align}
	\nn v_A^av_{B}^{b}\nablao_{(a}\lp \sone_{b)}\hspace{-0.02cm} +\hspace{-0.02cm}2r_{b)}\rp \hspace{-0.02cm} = \hspace{-0.02cm}\nabla^h_{(A} \lp \sone_{B)}\hspace{-0.02cm} + \hspace{-0.02cm} 2r_{B)}\rp \hspace{-0.02cm} + \hspace{-0.02cm} 2\Q\nabla^h_{(A}\ellp_{B)} \hspace{-0.02cm} - \hspace{-0.02cm} \lp \ellp^J\lp \sone_J \hspace{-0.02cm} + \hspace{-0.02cm} 2r_J \rp \hspace{-0.02cm} - \hspace{-0.02cm} 2\Q(\elltwo\hspace{-0.02cm}-\hspace{-0.02cm}\elltwo_{\parallel})\rp \U_{AB}.
\end{align}
Since $F_{Bc} n^c =  - \sone_{B}$ and $\F_{AB}= \nabla^h_{[A}\ellp_{B]}$ (by \eqref{FAB}), inserting \eqref{PUandPUU}-\eqref{trPYandtrPUwithPdec} into \eqref{defabsRicci} 
%\eqref{defabsRiccinull} 
yields
\begin{align}
	\nn \Rtensor_{AB} \defi  &\spc \Riemo_{(AB)}- 2 (\pounds_{{n}} \bY)_{AB}
	- \lp 2\Q +\textup{tr}_h\bUp \rp\Y_{AB}+ \nabla^h_{(A} \lp \sone_{B)}+2r_{B)}\rp+ 2\Q\nabla^h_{(A}\ellp_{B)}\\
	\nn & -2\Yn_A\Yn_B+4\Yn_{(A}\sone_{B)}-\sone_A\sone_B-\lp \textup{tr}_h\bYp-\Q(\elltwo-\elltwo_{\parallel})+ \ellp^J\sone_J   \rp \U_{AB}\\
	&+ 2 h^{CD}\U_{D(A}\lp 2\Y_{B)C}+\ellp_C\lp \sone_{B)}-2\Yn_{B)}\rp+\dfrac{1}{2}\lp \nabla^h_{B)}\ellp_C-\nabla^h_{C}\ellp_{\vert B)}\rp\rp.
\end{align}
Substituting expression \eqref{RiemoAB_1} for $\Riemo_{(AB)}$ and reorganizing terms, one easily arrives at \eqref{RconstraintABfinal}.
\end{proof}

\section{Gauge invariant quantities on a transverse submanifold $S$}\label{secRiemoAB:gauge:inv} 

Equation \eqref{RconstraintABfinal} is rather complicated, mainly because it has been written in a completely arbitrary gauge.\ This is clearly advantageous since the gauge can then be adjusted to the problem at hand.\ However, the equation involves a few quantities that are gauge invariant, namely
the constraint tensor $\Rtensor_{AB}$ and the metric $h_{AB}$ together with all its derived objects, such as the Levi-Civita covariant derivative $\nabla^h$ and
the Ricci tensor $R^h_{AB}$.\ A natural question arises as to whether one can find additional objects with simple gauge behaviour, so that one can write down
\eqref{RconstraintABfinal} fully in terms of gauge invariant quantities.\ There is an obvious answer to this, namely that the sum of all terms in the right-hand side of \eqref{RconstraintABfinal} except for the first one must necessarily be a gauge invariant quantity.\ While this must be true, it is clearly not very helpful.\ However, the idea behind it turns out to be useful.\ If we can find 
%simple 
gauge invariant quantities that can then be substituted in the equation, then the reminder must also be gauge invariant.\ This procedure can lead to the determination of gauge invariant objects that would have been very hard to guess otherwise.\ 
%Furthermore, showing explicitly that such object is indeed gauge invariant would provide a highly non-trivial independent test on the validity of equation \eqref{RconstraintABfinal}.
This is the task we set up to do in the present section.
% and in Appendix \ref{ApB}.

Based on the gauge behaviour discussed in Lemma \ref{Gauge_trans_UFsone} and Corollary \ref{Gr-s}, we  write down two quantities on $S$ with very simple gauge behaviour. The underlying reason why such objects behave in this way can be undestood  from the notion of normal pair and the associated geometric quantities on $S$ defined and studied in \cite{mars2023covariant}. For simplicity, however, here we simply put forward the definitions and find explicitly how they transform under an arbitrary change of gauge.
\begin{lemma}\label{gauge_invariant_1}
	In Setup \ref{setup}, define the covector $\bomegap$ and the symmetric $(0,2)$-tensor $\bsecp$ on $S$ by
	\begin{align}
		\label{def:omega:beta:gauge}\bomegap & \defi  \psi^{\star} ( \bsone - \bsYn) - \bUp ( \ell_{\parallel}^{\sharp}, \cdot), \qquad\bsecp\defi  \psi^{\star}  \bY + \frac{1}{2} \lp  \elltwo_{\parallel} - \elltwo|_S \rp\bUp- \frac{1}{2} \pounds_{\ell_{\parallel}^{\sharp}} h.
	\end{align}
	Under an arbitrary gauge transformation with gauge parameters $\{z,V\}$, they transform as
	\begin{align}
		\mathcal{G}_{\lp z,V\rp}( \bomegap ) = \bomegap - \frac{1}{\z} d \z, \qquad \quad\mathcal{G}_{\lp z,V\rp}( \bsecp ) = \z \bsecp,
	\end{align}
	where $\z \defi  \psi^{\star} z$.
\end{lemma}
\begin{proof}
	The proof is based on the results of Appendix \ref{sec:gauge-fix:res}, where the gauge behaviour of several hypersurface data quantities is derived.  In particular we use the definitions $\bw \defi  \gamma (V, \cdot)$
	and $\f \defi  \ellc ( V)$ introduced there.
	From \eqref{gaugeellbis} and \eqref{Uprime}  it follows
	(prime denotes a gauge-transformed object)
	\begin{align}
		\bellp' = \z (\bellp + \bwp ), \qquad \bUp'= \z^{-1} \bUp.
		\label{bellpprime}
	\end{align}
	Thus $\ell^A{}^{\prime} = \z ( \ell^A +  w^A )$ and 
	$\lp \U_{AB} \ell^B \rp^{\prime} = \U_{AB} \lp \ell^B +  w^B \rp$.\ 
	The transformation law of $\bomegap$ follows at once from this and Corollary
	\ref{Gr-s} (recall that $\bU(n,\cdot)=0$, $\gamma(n,\cdot)=0$). Concerning $\bsecp$ we use the decomposition $V^a = f n + P^{ab}  w_b$ (cf.\ \eqref{Vdecom}) 
	and apply Lemma \ref{pullbackLieT} to the transformation law \eqref{gaugeY} of $\bY$.\  This gives
	\begin{align*}
		\psi^{\star} \bY' =  \z \psi^{\star} \bY
		+ \bellp \otimes_s d \z + \z  \lp f |_S - \ell^C  w_C \rp
		\bUp  + \frac{1}{2} \pounds_{\z  w^{\sharp}} h.
	\end{align*}
	Since 
	\begin{align}
		\hspace{-0.2cm}\elltwo_{\parallel}{}^{\prime} = \z^2 \lp \elltwo_{\parallel} + 2 \ell^C  w_C
		+  w^C  w_C \rp, \quad  \elltwo{}^{\prime} |_S = \z^2 \lp \elltwo |_S + 2 f|_S
		+  w^C   w_C \rp, \label{l2l2parallelprime}
	\end{align}
	the first because of definition $\elltwo_{\parallel}\defi h^{\sharp}(\ellc_{\parallel},\ellc_{\parallel})$ and the second being a consequence of \eqref{gaugeell2bis} together with 
	%the expression of $P^{ab}$ obtained in 
	\eqref{Pdecom:abstract}, %\eqref{Pdecomp2},
	%Lemma \ref{lemPdec}, 
	one finds
	\begin{align*}
		\lp    (\elltwo_{\parallel} - \elltwo|_S ) \bUp \rp^{\prime}
		=  \z ( \elltwo_{\parallel} - \elltwo|_S ) \bUp 
		+ 2 \z \lp \ell^C  w_C - f |_S \rp \bUp.
	\end{align*}
	Given that $(\ell_{\parallel}^{\sharp})^{\prime} = \z  \ell_{\parallel}^{\sharp} + \z  w^{\sharp}$
	and $\pounds_{\z \ell_{\parallel}^{\sharp}} h = \z \pounds_{\ell_{\parallel}^{\sharp}} h
	+ 2 \bellp \otimes_s d \z$ all terms involving $ w^{\sharp}$ and $\tdo \z$ in
	$\bsecp^{\prime}$ cancel out and the transformation law
	$\bsecp^{\prime} = \z \bsecp$ follows.
\end{proof}
The result states, in particular, that $\bomegap$ and $\bsecp$ are nearly gauge invariant, and in fact that they are exactly gauge invariant under the subgroup
\begin{align*}
	{\mathcal G}_1 \defi  \{ 1, V \} \subset {\mathcal G} = \Fcal^{\star}(\N) \times \Gamma(T \N).
\end{align*}
The fact that ${\mathcal G}_1$ is a subgroup of ${\mathcal G}$ is immediate from the composition law of \tcb{$\G$ (see Section \ref{sec:gauge:structure:prelim})}.  %Proposition \ref{propGaugeGroup}.
%\eqref{composition}.

Following the idea outlined above, the next step is to rewrite the constraint tensor $\Rtensor$ on the submanifold $S$ in terms of these quantities.\  We still need to decide which objects of \eqref{RconstraintABfinal} are to be replaced. For $\bsecp$ there is only one natural choice, namely $\psi^{\star} \bY$. For $\bomegap$, we could replace either $\bm{\sone}$ or $\bsYn$, but the second choice is preferable because $\bomegap$ is not at the level of metric hypersurface data since it  involves some components of the tensor $\bY$ as well.

The following result is obtained by a simple computation
whereby $\bsYn$ and $\psi^{\star} \bY$ are replaced in terms of $\bomegap$ and $\bsecp$ respectively in \eqref{RconstraintABfinal}.
\begin{proposition}
	\label{GaugeInvForm}
	Assume Setup \ref{setup}.\ The pull-back to $S$ of the constraint tensor $\Rtensor$ reads
	\begin{align}
		\nn   \Rtensor_{AB} = & \spc R^h_{AB} - 2 \nabla^h_{(A} \omega_{B)}
		- 2 \omega_A \omega_B - \lp  2 \Q + \textup{tr}_h \bUp \rp
		\sec_{AB} \\
		&- (\textup{tr}_h \sec) \U_{AB} + 4 \sec^C{}_{(A} \U_{B)C} - 2  \thi_{AB},
		\label{RtensorABFinal_2}
	\end{align}
	where
	\begin{align}
		\nn \thi_{AB} \defi & \spc (\pounds_{{n}} \bY)_{AB} - \frac{1}{2} (\elltwo -\elltwo_{\parallel})(\pounds_{n}\bU)_{AB} -2 \nabla^h_{(A} \sone_{B)}
		\\
		&  - \lp \frac{1}{2} n(\elltwo) + \ellp^C\ellp^D\U_{CD} - 2 \ellp^{C}\sone_C \rp \U_{AB} + \ellp^C \lp -\nabla^h_C \U_{AB}+  2 \nabla^h_{(A} \U_{B) C} \rp.
		\label{defthi}
	\end{align}
\end{proposition}
The definition of the symmetric $(0,2)$-tensor $\bthip$ is not artificial.\ As mentioned above, the fact that the tensors $\psi^{\star}\Rtensor$ and $\Ricc^h$ are gauge invariant, together with the simple gauge behaviour of $\bomegap$,
$\bsecp$, $\bUp$ (cf.\ \eqref{Uprime}) and $\Q$ (cf.\ \eqref{Qprime}), imply that $\bthip$ must also have a simple gauge behaviour.\  
We emphasize that while the existence and explicit form of the
${\mathcal G}_1$-gauge invariant quantities $\bomegap$ and
$\bsecp$ can be justified by the use of normal pairs and their associated geometric objects \cite{mars2023covariant}, the existence of the  ${\mathcal G}_1$-gauge invariant quantity $\bthip$ could not be anticipated and comes as an interesting by-product of the constraint tensor.\  The tensor $\bthip$ contains information on the  first order variation of the extrinsic curvature $\bY$ along the null direction $n$.\ Defining this geometric object is one of the main results in this paper.

The quantity $\bthip$ has several interesting features that, in our opinion, deserve further
investigation. Here we shall only mention that this object is not only ${\mathcal G}_1$-gauge invariant and it has a simple full ${\mathcal G}$-gauge behaviour (which makes it computable in any gauge)  but it is also intrinsic to the submanifold $S$.\ By ``intrinsic" we mean that it encodes geometric information of $S$ as a submanifold of $\N$ (or of the ambient space $({\mathcal M},g)$ in case the data is embedded), independently of $S$ belonging or not to any foliation of $\N$.\ This information is at the level of second derivatives (curvature), unlike $\bomegap$ or $\bsecp$ which involve only first derivatives (extrinsic curvature).

The gauge behaviour of $\bthip$ is obtained next as a consequence of Proposition
\ref{GaugeInvForm}. 
\begin{corollary}
	\label{invthi}
	Under a gauge transformation with gauge parameters $\{z,V\}$, 
	$\bthip$ transforms as
	\begin{align}
		{\mathcal G}_{(z,V)} \lp  \thi \rp_{AB}  =
		\thi_{AB}
		+ \frac{1}{\z} \nabla^h _A \nabla^h_B \z
		- \frac{2}{\z^2} \nabla^h_A z \nabla^h_B \z
		+ \frac{2}{\z} \omega_{(A} \nabla^h_{B)} \z
		+ \frac{\z_n}{\z} \sec_{AB},
		\label{gaugeTransrp}
	\end{align}
	where $\z\defi  z |_S$ and $\z_n\defi  n(z)|_S$. In particular $\bthip$ is invariant under the subgroup ${\mathcal G}_1$.
\end{corollary}
\begin{proof}
	We apply a  gauge transformation with gauge parameters $\{z,V\}$ to \eqref{RtensorABFinal_2} and subtract the equation itself.\ Using a prime to denote gauge
	transformed objects, one finds
	\begin{align*}
		0 = -2 \nabla^h_{(A} \lp \omega'_{B)} - \omega_{B)} \rp
		- 2 \omega'_A \omega'_B + 2 \omega_A \omega_B
		- 2 \left ( \Q' \z - \Q \right ) \sec_{AB} 
		-2 \thi'_{AB} + 2 \thi_{AB},
	\end{align*}
	where we  used the gauge invariance of $\psi^{\star} \Rtensor$, $h$, $\nabla^h$ and $\Ricc^h$, as well as the fact that $\bUp$ scales with $\z^{-1}$ (cf.\ \eqref{Uprime}) while $\bsecp$ scales with $\z$ (so their product is gauge invariant).\ Using the definition
	$\z_n\defi  n(z)|_S$ and inserting $\bomegap' = \bomegap - \z^{-1} \tdo \z$ as well as \eqref{Qprime}, the result follows after simple cancellations. 
\end{proof}
\begin{remark}
	As already said, guessing that $\bthip$ has a simple gauge behaviour without the aid of Proposition \ref{GaugeInvForm} appears to be hard.\ Conversely, checking explicitly that $\bthip$ has the gauge behaviour described in Corollary
	\ref{invthi} serves as a stringent consistency test for the validity
	of Proposition \ref{GaugeInvForm}.\ It is worth mentioning that the computation that proves \eqref{gaugeTransrp} by direct calculation is somewhat long and will not be included here.\ For details in this regard we refer to \textup{\cite{manzano2023PhD}}.
\end{remark}

\section{A generalized near horizon equation}\label{sec:GME}

To illustrate the importance of the quantities $\{\bomegap,\bsecp,\bthip\}$ and the result \eqref{RtensorABFinal_2}, let us consider the case of a totally geodesic null hypersurface.\ In the recent work \cite{manzano2023field}, we have introduced three new notions of Killing horizons that are \textit{purely abstract}, meaning that they have been constructed without requiring the horizon to be embedded in any ambient space.\ We have also proved that these notions generalize the standard concepts of non-expanding, weakly isolated and isolated horizons \cite{ashtekar2000generic,ashtekar2002geometry,ashtekar2000isolated,gourgoulhon20063+,jaramillo2009isolated,krishnan2002isolated} to completely general topologies and to horizons (possibly) with fixed points.\  
%(i.e.\ points of the horizon where the symmetry generator vanishes).\ 
%
One of these notions is that of 
%In particular, an 
\textit{abstract Killing horizon of order zero}, which is null hypersurface data $\hypdata$ admitting a gauge-invariant vector field $\ovkil\in\text{Rad}\gamma$ so that $(i)$ the subset $\mathcal{S}\defi \{p\in\N\spc\vert\spc\kil\vert_p=0\}$ has empty interior and $(ii)$ $\bU=0$ \cite[Def.\ 6.1]{manzano2023field}.\

For our purposes here, let us exclude the possibility of $\ovkil$ vanishing on $\N$ and consider the case when the horizon $\hypdata$ is $\{\phi,\rig\}$-embedded in a spacetime $(\M,g)$.\ In these circumstances, the hypersurface $\phi(\N)$ is totally geodesic\footnote{Recall that $\bU$ is the second fundamental form of $\phi(\N)$ w.r.t.\ the null normal  $\phi_{\star}n$.}.\  
%abstract Killing horizons of order zero are therefore totally geodesic, as $\bU$ coincides with the second fundamental form of the hypersurface in the embedded case.\ 
%, and consider the case when the horizon $\hypdata$ is embedded in a spacetime $(\M,g)$.\ 
%For simplicity, .\ 
Since $\ovkil$ belongs to the radical of $\gamma$, it follows that $\ovkil=\alpha n$ for a suitable (nowhere zero) function $\alpha\in\Fcal^{\star}(\N)$.\ Furthermore, 
%combining 
the gauge-invariance of $\ovkil$ and the gauge behaviour \eqref{gaugen} of $n$ allow one to
%one can 
select the gauge so that $\ovkil=n$.\ Following the setup and notation in \cite{manzano2023field}, we let $\kil$ be the extension of $\phi_{\star}\ovkil$ off $\phi(\N)$ and define the quantities %\cite{manzano2023field}
\begin{align}
	\label{defwpqb}
	%\p&\defi \phi^{\star}(\Kkil \left (\rig,  \rig  \right )),\quad
	 \w\defi \phi^{\star}\big((\Kkil) \left (\rig,  \nu  \right )\hspace{-0.05cm}\big),\quad  \bqone\defi \phi^{\star}\big((\Kkil) \left (\rig, \cdot \right )\hspace{-0.05cm}\big),\quad
	%	\\
	%	\label{defwpqb:Y}
	\bcalY\defi \dfrac{1}{2}\phi ^{\star}\big ( \pounds_{\rig} \Kkil \big ).
\end{align}
%\begin{align}
%	\label{defwpqb:Y}\bcalY&\defi \dfrac{1}{2}\phi ^{\star}\left ( \pounds_{\rig} \Kkil \right ).
%\end{align}
%
%
A key result in \cite{manzano2023field} is equation (5.15), which establishes a relation between the Lie derivative of $\bY$ along $n$ and 
%the components of 
the deformation tensor $\Kkil$ of $\kil$ on a general null hypersurface admitting a null and tangent privileged vector field.\ Setting 
%Particularizing this identity to the case when 
$\bU=0$ and $\alpha=1$ in such equation gives % in such identity transforms it into %, it becomes %, such identity reads
%%\cite[Eq.\ (5.15)]{manzano}
\begin{align}
	\label{Lie:id:sigma} 
	0=\pounds_n\Y_{ab} -2\nablao_{( a}\sone_{b)}+\nablao_{(a}\qone_{b)}+\w\Y_{ab}-\calY_{ab}.
\end{align}
We can now compute the pull-back of \eqref{Lie:id:sigma}  to a transverse submanifold $S$ of $\N$.\ Assuming Setup \ref{setup} and using the identities 
%Using that 
(recall \eqref{covderpcovtensoronS} and note that $\bqone(n)=\w$)
\begin{align}
2v_A^av_B^b\nablao_{(a}\sone_{b)}&=2\nabh_{(A}\sone_{B)},\qquad %\\
%v_{A}^{a}v_B^b\nablao_b\qone_{a}&= \nabla^h_{B} \qone_{A}-\w \nabla^h_{(A}\ellp_{B)}.\\
%v_{A}^{a}v_B^b\nablao_a\qone_{b}&= \nabla^h_{A} \qone_{B}-\w \nabla^h_{(A}\ellp_{B)}.\\
v_{A}^{a}v_B^b \nablao_{(a}\qone_{b)}=  \nabla^h_{(A} \qone_{B)}-\w \nabla^h_{(A}\ellp_{B)},
\end{align}
together with the definition of $\bsecp$ in \eqref{def:omega:beta:gauge}, 
%\begin{align}
%\sec_{AB}\defi  \Y_{AB} -  \nabh_{(A}\ell_{B)}.
%\end{align}
it is straightforward to find 
%the pull-back to $S$ of \eqref{Lie:id:sigma}, which is 
\begin{align}
	\label{Lie:id:sigma:pullback} 
	(\pounds_n\bY)_{AB}
	%= 2\nabh_{(A}\sone_{B)}-\nabla^h_{(A} \qone_{B)}-\w\lp \Y_{AB}- \nabla^h_{(A}\ellp_{B)}\rp+\calY_{AB}\\
	= 2\nabh_{(A}\sone_{B)}-\nabla^h_{(A} \qone_{B)}-\w\sec_{AB}+\calY_{AB}.
\end{align}
Equation \eqref{Lie:id:sigma:pullback} provides the derivative $(\pounds_n\bY)_{AB}$ in terms of hypersurface data and the quantities $\{\w,\bqone,\bcalY\}$ introduced in \eqref{defwpqb}.\ These quantities depend on the extension of $\kil\vert_{\phi(\N)}$ off $\phi(\N)$, and in particular they vanish identically whenever $\kil$ happens to be a Killing vector on $(\M,g)$.\ Inserting \eqref{Lie:id:sigma:pullback} into the definition \eqref{defthi} of $\bthip$ yields %($\bU=0$)
\begin{align}
	\thi_{AB}= & \spc (\pounds_{{n}} \bY)_{AB}  -2 \nabla^h_{(A} \sone_{B)}=-\nabla^h_{(A} \qone_{B)}-\w\sec_{AB}+\calY_{AB},
\end{align}
%\begin{align}
%	\nn -2\thi_{AB}=+2\nabla^h_{(A} \qone_{B)}+2\w\sec_{AB}-2\calY_{AB},
%\end{align}
which in turn transforms \eqref{RtensorABFinal_2} into
%	\begin{align}
%	\nn   \Rtensor_{AB} = & \spc R^h_{AB} - 2 \nabla^h_{(A} \omega_{B)}
%	- 2 \omega_A \omega_B +2\lp \w- \Q 
%	\rp\sec_{AB}+2\nabla^h_{(A} \qone_{B)}-2\calY_{AB},
%\end{align}
\begin{align}
	 0=\frac{1}{2}\lp \Rtensor_{AB}-R^h_{AB} \rp+  \nabla^h_{(A} \omega_{B)}+ \omega_A \omega_B+\lp \Q -\w 
	\rp\sec_{AB}-\nabla^h_{(A} \qone_{B)}+\calY_{AB}.
	\label{near:hor:eq:U=0}
\end{align}
Identity \eqref{near:hor:eq:U=0} holds for \textit{any} totally geodesic embedded null hypersurface $\phi(\N)$ admitting an everywhere non-zero null and tangent \textit{gauge-invariant} vector field $\kil\vert_{\phi(\N)}$.\ The construction, in addition, can be adjusted to \textit{any possible extension of  $\kil\vert_{\phi(\N)}$ off $\phi(\N)$}.\
In order to endow \eqref{near:hor:eq:U=0} with a geometric interpretation, we first need to know the meaning of the objects $\Q$, $\bomegap$ and $\bsecp$.\ 
The function $\Q$ defined in \eqref{defY(n,.)andQ} is the surface gravity of $\phi_{\star}n$, since %because %(note that $\nablao_nn=0$, cf.\ \eqref{nablaonnull})
\begin{equation}
	\nabla_{\phi_{\star}n}\phi_{\star}n\stackbin{\eqref{nablaXYnablao}}=\phi_{\star}\Big(\nablao_{n}n-\bY(n,n)
	n\Big) \stackbin{\eqref{nablaonnull}}=\Q\phi_{\star}n. 
\end{equation}
On the other hand, when the rigging vector $\rig$ is taken null and transverse to $\phi\big(\psi(S)\big)$, a straightforward calculation based on \eqref{YtensorEmbDef} allows one to conclude that $\bs{r}_{\parallel}$ and $\bYp$ coincide with the torsion one-form 
%of $\psi(S)$ 
and the second fundamental form of $\psi(S)$ w.r.t.\ 
%the rigging 
$\rig$.\  
%respectively.\ 
The quantities $\bomegap$ and $\bsecp$ can therefore be understood as $\G_1$-invariant versions of these tensor fields, defined for a more general rigging.\ For all these reasons, \eqref{near:hor:eq:U=0} constitutes a generalization of the standard near horizon equation of isolated horizons \cite{ashtekar2000generic,ashtekar2002geometry,ashtekar2000isolated,gourgoulhon20063+,jaramillo2009isolated,krishnan2002isolated},  
%(or the master equation of multiple Killing horizons \cite{}), 
and extends it to the more general case of a totally geodesic null hypersurface with arbitrary topology, and with a privileged generator $\kil\vert_{\phi(\N)}$ which can be extended arbitrarily off $\phi(\N)$.\ In fact, in a forthcoming work \cite{manzano2023master} we shall obtain a new version of equation \eqref{near:hor:eq:U=0}, valid for \textit{completely general} null hypersurfaces equipped with a gauge-invariant tangent null vector.

\section{Construction of the tensor $\bY$ from initial data on a cross-section}\label{sec:application}

In Section \ref{secAbstractRicci}, we emphasized the importance of defining the constraint tensor $\Rtensor$ so that the tensor $\bY$ appears explicitly.\ We also stressed that writing the identities in terms of the Lie derivative of $\bY$ along $n$ could be advantageous in many circumstances.\ Our purpose here is to demonstrate that, given null \textit{metric} hypersurface data $\metdata$ 
%, the constraint tensor $\Rtensor$, a function $f\in\Fcal(\N)$ and 
plus suitable initial data on a cross-section $S$ of $\N$, the results of Section \ref{secAbstractRicci} allow for the construction of a hypersurface data tensor $\bY$ everywhere on $\N$.\ In order to prove this, we shall make use of a reasoning similar to the one in \cite[Theorem 4.7]{mars2024transverseII}, where from $\Lambda$-vacuum null metric data\footnote{Namely null metric hypersurface data with constraint tensor $\Rtensor=\Lambda \gamma$, $\Lambda\in\mathbb{R}$.} 
%admitting a cross-section $S$ 
plus certain initial data on a  cross-section $S$ 
%involving a sequence of traceless symmetric tensors on $S$, 
the authors prove existence of a semi-Riemannian manifold where the data is embedded so that the $\Lambda$-vacuum equations are satisfied to infinite order at the hypersurface.
%This is done in the following proposition.
%
Note, however, that in the next proposition we do not restrict ourselves to $\Lambda$-vacuum and the result is stated at the purely abstract level.\ 
\begin{proposition}\label{proposition:integrate:Y:tensor}
Consider null metric hypersurface data $\metdata$ admitting 
a cross-section $S$ (i.e.\ 
a transverse submanifold $S$ intersected precisely once by each integral curve of $n$), and assume Setup \ref{setup}.\ Let $\bs{\beta}$, $\bs{\mathcal{Y}}$ be a covector and a symmetric $(0,2)$-tensor on $S$, and $f\in\Fcal(\N)$, $\R_{ab}$ be respectively a scalar  function and a symmetric $(0,2)$-tensor verifying
%\miguel{No vamos a usar $\Rtensor$ todavia, va a ser $\Rtensor(n,n)=A$, ...} 
\begin{equation}
	\label{eq:for:rho:v1}\R_{ab}n^an^b = -n(\textup{tr}_P\bU)+(\trP\bU) f -P^{ab}P^{cd}\U_{ac}\U_{bd}\spc\spc\text{on}\spc\spc\N.
\end{equation}
%on $\N$.\ 
Then, there exists a unique symmetric $(0,2)$-tensor field $\bY$ on $\N$  satisfying 
\begin{align}
\label{requirements:Yconstruction}	\Q\defi-\bY(n,n)=f\spc\spc\text{on}\spc\spc \N,\qquad\quad and\qquad\quad \psi^{\star}\big(\bY(n,\cdot)\big) =\bs{\beta},\spc\spc \psi^{\star}\bY=\bs{\mathcal{Y}}\spc\spc\text{on}\spc\spc S,
\end{align}
such that 
%[components of $\bY$ coincide with $\beta$, $\mathcal{Y} $ y lo que sea en $\psi(S)$], and such that 
the constraint tensor of the hypersurface data $\{\N,\gamma,\ellc,\elltwo,\bY\}$ 
%, defined by \eqref{defabsRicci}, 
agrees with $\R_{ab}$. 
\end{proposition}
\begin{proof}
The idea is to construct, from 
%the initial data 
$\{\bs{\beta}, \bs{\mathcal{Y}}\}$ and %the function 
$f$, a covector $\mathbbm{r}$ and a symmetric $(0,2)$-tensor $\mathbb{Y}$ on $\psi(S)$, and then extend them off $\psi(S)$ by using suitable transport equations along the generators of $\N$.\ We then prove that $\mathbbm{r}(n)=-f$ and  $\mathbb{Y}(n,\cdot)=\mathbbm{r}$ everywhere on $\N$, as well as that $\mathbb{Y}$ is symmetric.\ This allows one to construct hypersurface data from $\metdata$ and $\mathbb{Y}$, and check that the constraint tensor of such data precisely coincides with $\R_{ab}$.\ 
It is important to introduce a quantity $\mathbbm{r}$ that is a priori unrelated to $\mathbb{Y}(n,\cdot)$ because this allows one to analyze the constraint tensor equations as ODEs, instead of as  PDEs.\ This is why part of the argument is to show that, a posteriori,  $\mathbb{Y}(n,\cdot)=\mathbbm{r}$.
%Of course part of the argument will be to show that, a posteriori, $\mathbb{Y}(n,\cdot )= \mathbbm{r}$.
	
%Recall that the metric hypersurface connection $\nablao$ only requires metric hypersurface data to be defined.\ Likewise, the Ricci tensor $\Riemo_{ab}$ can be constructed from the data fields $\{\gamma,\ellc,\elltwo\}$.\ 

Let $\mathbbm{r}_b$, $\mathbb{Y}_{ab}$ be the (unique) solutions of the following ODEs  along the integral curves of $n$: %transport equations
\begin{align}
\label{eq:for:rhoa:3456}\R_{ab}n^a =&  -\nablao_{b}f  - \pounds_n(\mathbbm{r}_b-\sone_{b})  - (\trP\bU)\lp \mathbbm{r}_{b}  - \sone_b\rp  - \nablao_b(\textup{tr}_P\bU) + P^{cd} \left ( \nablao_c\U_{bd} - 2 \U_{bd}\sone_{c} \right ),\\
\nn \R_{ab} =  &\spc \Riemo_{(ab)}- 2 \pounds_{{n}} \mathbb{Y}_{ab}
	- \lp 2f+\trP\bU \rp\mathbb{Y}_{ab}+ \nablao_{(a} \lp \sone_{b)}+ 2  \mathbbm{r}_{b)}\rp  \\
\label{defabsRicci:2} & -2\mathbbm{r}_a\mathbbm{r}_b+4\mathbbm{r}_{(a}\sone_{b)}-s_as_b-(\trP\mathbb{Y})\U_{ab}+ 2P^{cd}\U_{d(a}\lp 2\mathbb{Y}_{b)c}+\F_{b)c}\rp,
\end{align}
 with initial data 
\begin{equation}
\nn 	-\mathbbm{r}(n)\stackbin{\psi(S)}=f,\quad\spc  \mathbbm{r}(\psi_{\star}X)\stackbin{\psi(S)}=\bs{\beta}(X), \quad\spc \mathbb{Y}(n,\cdot)\stackbin{\psi(S)}=\mathbbm{r},
	\quad\spc\mathbb{Y}(\cdot,n)\stackbin{\psi(S)}=\mathbbm{r},
	\quad\spc \mathbb{Y}(\psi_{\star}X,\psi_{\star}Z)\stackbin{\psi(S)}=\bs{\mathcal{Y}}(X,Z)
\end{equation}
for all $X,Z\in\Gamma(TS)$.\  
The tensor  $\mathbb{Y}_{ab}$ is symmetric as a consequence of 
%\eqref{defabsRicci:2} and 
$\R_{ab}$ being symmetric.\ Indeed, \eqref{defabsRicci:2} implies  
\begin{equation}
	0=  - 2 \pounds_{{n}} \big(\mathbb{Y}_{ab}-\mathbb{Y}_{ba}\big)
	- \lp 2f+\trP\bU \rp\big(\mathbb{Y}_{ab}-\mathbb{Y}_{ba}\big),
\end{equation}
which combined with $\mathbb{Y}_{ab}=\mathbb{Y}_{ba}$ on $\psi(S)$ establishes the symmetry of $\mathbb{Y}$ everywhere on $\N$.\ Contracting \eqref{eq:for:rhoa:3456} with $n^b$ yields
\begin{equation}
	\label{eq:for:rho:v2}\R_{ab}n^an^b=  -n(f +\mathbbm{r}(n)\big)  - (\trP\bU)\hspace{0.05cm}\mathbbm{r}(n)  - n(\textup{tr}_P\bU) - P^{bf}P^{cd}\U_{bd}\U_{cf} 
\end{equation}
after using $\bsone(n)=0$, $\bU(n,\cdot)=0$ and $n^b\nablao_c\U_{bd}=-\U_{bd}\nablao_cn^b=-P^{bf}\U_{bd}\U_{cf}$ (cf.\ \eqref{nablaonnull}).\ Subtracting \eqref{eq:for:rho:v1} from \eqref{eq:for:rho:v2} then gives the ODE
\begin{equation}
	\label{eq:for:n(f+r(n))} 0 =  n\big(f +\mathbbm{r}(n)\big)  + (\trP\bU)\big(\mathbbm{r}(n)  + f\big)
\end{equation}
along the generators of $\N$.\ This equation, together with the initial condition $\mathbbm{r}(n)=-f$ on $\psi(S)$, guarantees that $\mathbbm{r}(n)=-f$ everywhere on $\N$.\ 
Our next aim is to show that $\mathbb{Y}(n,\cdot) = \mathbbm{r}$ holds everywhere, not just on the initial section $S$.\ 
%An analogous result can be obtained for $\mathbb{Y}$.\ 
Indeed, combining \eqref{Riemosym(n,-)}, $\bU(n,\cdot)=0$, $\bsone(n)=0$, and $\mathbbm{r}(n)=-f$ with the identity (recall \eqref{contrNsym}) $$n^b \nablao_{(b}\lp \sone_{a)}+ 2  \mathbbm{r}_{a)}\rp=\frac{1}{2}\pounds_{n}\lp \sone_{a}+ 2  \mathbbm{r}_{a}\rp    -\nablao_{a}f  +  2f \sone_a   -   P^{cb}\lp \sone_{b}+ 2  \mathbbm{r}_{b}\rp\U_{ca},$$ one finds that the contraction of \eqref{defabsRicci:2} with $n^b$ is
	\begin{align}
		\nn \R_{ab}n^b =  & - 2 \pounds_{{n}} \big(\mathbb{Y}_{ab}n^b\big)
		- \lp 2f+\trP\bU \rp\mathbb{Y}_{ab}n^b    +2f \mathbbm{r}_a -\nablao_{a}f +\pounds_{n}\lp \sone_{a}+    \mathbbm{r}_{a}\rp\\
		& -\nablao_a(\textup{tr}_P\bU)+ (\textup{tr}_P\bU)\sone_a  +P^{cd}\lp\nablao_c\U_{ad}  + 2 \U_{ca} \big( \mathbb{Y}_{bd}n^b  -\sone_{d}  -    \mathbbm{r}_{d}\big) \rp.
		\label{defabsRicci:4} 
	\end{align}
	Subtracting \eqref{eq:for:rhoa:3456} from \eqref{defabsRicci:4}, one obtains the following ODE for $\mathbb{Y}_{ab}n^b-\mathbbm{r}_a$:
	%\begin{equation}
	%	\rho_a =  -\nablao_{a}f  - \pounds_n(\mathbbm{r}_a-\sone_{a})  - (\trP\bU)\lp \mathbbm{r}_{a}  - \sone_a\rp  - \nablao_a(\textup{tr}_P\bU) + P^{cd} \left ( \nablao_c\U_{ad} - 2 \U_{ca}\sone_{d} \right ) 
	%\end{equation}
	\begin{align}
		0=2 \pounds_{{n}} \big(\mathbb{Y}_{ab}n^b-\mathbbm{r}_a\big)+\lp   \big(2f +\trP\bU \big)\delta_a^d - 2P^{cd}  \U_{ca}\rp \big( \mathbb{Y}_{bd}n^b   -    \mathbbm{r}_{d}\big)  , 
	\end{align}
	which together with the fact that 
	%Once more, since 
	$\mathbb{Y}(n,\cdot)=\mathbbm{r}$ on $\psi(S)$ guarantees that 
	%, it follows that 
	$\mathbb{Y}(n,\cdot)=\mathbbm{r}$ everywhere on $\N$.\ 
%    \tcr{The tensor  $\mathbb{Y}_{ab}$ is symmetric as a consequence of 
%	%\eqref{defabsRicci:2} and 
%	$\R_{ab}$ being symmetric.\ Indeed, \eqref{defabsRicci:2} implies  
%		\begin{equation}
%	0=  - 2 \pounds_{{n}} \big(\mathbb{Y}_{ab}-\mathbb{Y}_{ba}\big)
%	- \lp 2f+\trP\bU \rp\big(\mathbb{Y}_{ab}-\mathbb{Y}_{ba}\big),
%	\end{equation}
%	which combined with $\mathbb{Y}_{ab}=\mathbb{Y}_{ba}$ on $\psi(S)$ establishes the symmetry of $\mathbb{Y}$ everywhere on $\N$.\ }
	Since $\mathbb{Y}$ is symmetric, this  
	%
	%This 
	means
	%, in particular, 
	that $\{\N,\gamma,\ellc,\elltwo,\bY\defi\mathbb{Y}\}$ constitutes null hypersurface data.\ The tensor $\bY$ 
	%At this stage, we can define a symmetric $(0,2)$-tensor $\bY$ on $\N$ by $\bY\defi\mathbb{Y}$.\ This tensor field 
	satisfies \eqref{requirements:Yconstruction} by construction.\   %$$\Q\defi-\bY(n,n)=-\mathbb{Y}(n,n)=f,\quad\text{and}\quad \bs{r}\defi\bY(n,\cdot)=\mathbb{Y}(n,\cdot)=\mathbbm{r}.$$
%	\begin{equation}
%	\bY(n,n)=-f,\quad\bY(n,\cdot)=\mathbbm{r},\quad\bY_{ab}=
%	\end{equation}
	Moreover, by \eqref{defabsRicci:2} it holds that $\R_{ab}$ coincides with the constraint tensor of $\{\N,\gamma,\ellc,\elltwo,\bY\}$ (cf.\ \eqref{defabsRicci}). %Therefore, the tensor $\bY=\mathbb{Y}$ satisfies \eqref{defabsRicci}. 
	Uniqueness follows from the fact that we are solving ODEs with prescribed initial data.\
\end{proof}
Observe that, in the case when $\N$ does not admit a cross-section,  Proposition \ref{proposition:integrate:Y:tensor} still holds locally near any transverse submanifold $S$, so one can still construct a unique hypersurface data tensor field $\bY$ in a neighbourhood of $S$.\ 
In general the solutions in different patches will not combine well to define a tensor field $\bY$ globally on $\N$.\ This property is guaranteed when $\N$ has product topology, but not in other cases.\  
For instance, if $\N$ had topology $S\times\mathbb{S}^1$, $\mathbb{S}^3$ with the integral curves of $n$ along the Hopf fibration, or $\mathbb{T}^3$,  the condition of 
%the global tensor 
$\bY$ being everywhere smooth on $\N$ would not follow directly from the integration of the ODEs above, and would require a careful analysis of compatibility of solutions along different generators when they are either closed ($S\times \mathbb{S}^1$ and Hopf fibration cases) or become infinitely close to each other ($\mathbb{T}^3$ case).

\appendix

\section{A generalized Gauss identity}\label{secGauss}
In this appendix we present a generalization of the well-known Gauss identity (see e.g.\ \cite{kobayashi1969vol2}).\ On any semi-Riemannian manifold, the Gauss identity relates  the curvature tensor of the Levi-Civita connection  along tangential directions of a non-degenerate hypersurface with the curvature tensor of the induced metric and the second fundamental form. It has been generalized in a number of directions, e.g.\ for the induced connection associated to  a transverse (rigging) vector \cite{mars1993geometry}. Here we find an identity where the connection of the space and of the hypersurface are completely general, except for the condition that they are both torsion-free.
 
Our primary interest will be in applying this identity to
a non-degenerate
codimension-one submanifold of a space endowed with null hypersurface data.\ 
However, the identity is  more general and may be of independent interest.\ We remark that the tensor $\widehat{\gamma}$ in the statement of Theorem
\ref{thmGeneralIdentity}
is  arbitrary, so neither $\widehat{\gamma}$ nor its pull-back $\widehat{h}$ to the submanifold are assumed to be non-degenerate.
\begin{theorem}\label{thmGeneralIdentity}
  Consider a smooth manifold $\N$ endowed with a symmetric $(0,2)$-tensor field $\widehat{\gamma}$ and a torsion-free connection $\whnabla$.\ Let
$\psi:S\longhookrightarrow\N$ be an embedded hypersurface in $\N$ and assume that $S$ is equipped with another torsion-free connection $\whD$.\ Define $\widehat{h}\defi\psi^{\star}\widehat{\gamma}$ 
  and the tensor $\mathcal{P}$ by means of
\begin{align*}
\widehat{\nabla}_XY=\widehat{D}_XY+\mathcal{P}(X,Y)\qquad\forall X,Y\in\Gamma(TS),
\end{align*}
and assume that there exists a transversal vector field $n$ along
$S$ satisfying $\psi^{\star}(\widehat{\gamma}(n,\cdot)) =0$.
Define the $(0,2)$-tensor $\Omega$ and the $(1,2)$-tensor $A$ on $S$ by
decomposing $\mathcal{P}(X,Y)$ in a tangential and a transverse part along
$n$ as follows:
\begin{align}
\label{Sdecomp}\mathcal{P}(X,Y) = A(X,Y) + \Omega(X,Y) n .
\end{align}
Then, for all $X,Y,Z,W\in\Gamma(TS)$ it holds 
\begin{align}
\nn \widehat{\gamma}(W,R^{\whnabla}(X,Y)Z)=&\spc \widehat{h}(W,R^{\whD}(X,Y)Z)+(\whD_XA_{\widehat{h}})(W,Y,Z)-(\whD_YA_{\widehat{h}})(W,X,Z)\\
\nn & +\widehat{h}(A(Y,W),A(X,Z))-\widehat{h}(A(X,W),A(Y,Z))\\
\nn &-(\whnabla_X\widehat{\gamma})(W,\mathcal{P}(Y,Z))+(\whnabla_Y\widehat{\gamma})(W,\mathcal{P}(X,Z))\\
\label{mideq22}& +\widehat{\gamma}(n,n)\lp\Omega(Y,W)\Omega(X,Z)-\Omega(X,W)\Omega(Y,Z)\rp,
\end{align}
where $A_{\widehat{h}}(W,X,Z)\defi\widehat{h}(W,A(X,Z))$.
\end{theorem}
\begin{proof}
The connections are torsion-free, so  $\mathcal{P}(X,Y)$, $A(X,Y)$, 
$\Omega(X,Y)$ are all symmetric in $X$, $Y$. We start by computing
\begin{align}
\label{mideq28}\widehat{\nabla}_X\widehat{\nabla}_YZ&=\whnabla_X(\whD_YZ+\mathcal{P}(Y,Z))=\whD_X\whD_YZ+\mathcal{P}(X,\whD_YZ)+\whnabla_X(\mathcal{P}(Y,Z)),\\
\label{mideq29}\widehat{\nabla}_{[X,Y]}Z&=\whD_{[X,Y]}Z+\mathcal{P}(\whD_XY,Z)-\mathcal{P}(\whD_XY,Z).
\end{align} 
The quantity $(\mathscr{D}_X\mathcal{P})(Y,Z)\defi\whnabla_X(\mathcal{P}(Y,Z))-\mathcal{P}(\whD_XY,Z)-\mathcal{P}(Y,\whD_XZ)$ is tensorial in $X,Y,Z$, and takes values in the space of vector fields (not necessarily tangent) along $\psi(S)$. Inserting \eqref{mideq28}-\eqref{mideq29} into the definition of 
%the 
curvature tensor 
%\eqref{curvoperator} 
yields
\begin{equation}
\nn R^{\whnabla}(X,Y)Z=R^{\whD}(X,Y)Z+(\mathscr{D}_X\mathcal{P})(Y,Z)-(\mathscr{D}_Y\mathcal{P})(X,Z).
\end{equation}
We now insert the decomposition \eqref{Sdecomp}. Using that $\widehat{\gamma}(n,W)=0$ and $\widehat{h}(X,Y)\defi\widehat{\gamma}(X,Y)$ gives
\begin{align}
\nn \hspace{-0.2cm}\widehat{\gamma}(\whnabla_XW,\mathcal{P}(Y,Z))%=&\spc\widehat{\gamma}(\whD_XW+A(X,W)+\Omega(X,W)n,A(Y,Z)+ \Omega(Y,Z)n)\\
 =&\spc\widehat{\gamma}(\whnabla_XW,A(Y,Z)+\Omega(Y,Z)n)\\
\label{mideq30} =&\spc\widehat{h}(\whD_XW,A(Y,Z))+\widehat{h}(A(X,W),A(Y,Z))+\widehat{\gamma}(n,n)\Omega(X,W)\Omega(Y,Z),
\end{align}
from where it follows
\begin{align}
\nn \widehat{\gamma}(W,(\mathscr{D}_X\mathcal{P})(Y,Z))=&\spc \widehat{\gamma}(W,\whnabla_X(\mathcal{P}(Y,Z)))-\widehat{\gamma}(W,\mathcal{P}(\whD_XY,Z))-\widehat{\gamma}(W,\mathcal{P}(Y,\whD_XZ))\\
\nn =&\spc\whnabla_X\lp \widehat{\gamma}(W,\mathcal{P}(Y,Z))\rp-(\whnabla_X\widehat{\gamma})(W,\mathcal{P}(Y,Z))-\widehat{\gamma}(\whnabla_XW,\mathcal{P}(Y,Z))\\
\nn & -\widehat{\gamma}(W,\mathcal{P}(\whD_XY,Z))-\widehat{\gamma}(W,\mathcal{P}(Y,\whD_XZ))\\
\nn \stackbin{\eqref{mideq30}}=&\spc\whD_X\lp \widehat{h}(W,A(Y,Z))\rp-(\whnabla_X\widehat{\gamma})(W,\mathcal{P}(Y,Z))-\widehat{h}(\whD_XW,A(Y,Z))\\
\nn &-\widehat{h}(A(X,W),A(Y,Z))-\widehat{\gamma}(n,n)\Omega(X,W)\Omega(Y,Z)\\
\nn &-\widehat{h}(W,A(\whD_XY,Z))-\widehat{h}(W,A(Y,\whD_XZ))\\
\nn =&\spc(\whD_X\widehat{h})(W,A(Y,Z))-(\whnabla_X\widehat{\gamma})(W,\mathcal{P}(Y,Z))+\widehat{h}(W,(\whD_XA)(Y,Z)))\\
\nn & -\widehat{h}(A(X,W),A(Y,Z))-\widehat{\gamma}(n,n)\Omega(X,W)\Omega(Y,Z).
\end{align}
Therefore,
\begin{align}
\nn \widehat{\gamma}(W,R^{\whnabla}(X,Y)Z)=&\spc \widehat{h}(W,R^{\whD}(X,Y)Z)+(\whD_X\widehat{h})(W,A(Y,Z))-(\whnabla_X\widehat{\gamma})(W,\mathcal{P}(Y,Z))\\
\nn &+\widehat{h}(W,(\whD_XA)(Y,Z)))-\widehat{h}(A(X,W),A(Y,Z))-(\whD_Y\widehat{h})(W,A(X,Z))\\
\nn & +(\whnabla_Y\widehat{\gamma})(W,\mathcal{P}(X,Z))-\widehat{h}(W,(\whD_YA)(X,Z)))+\widehat{h}(A(Y,W),A(X,Z))\\
\label{mideq21}& +\widehat{\gamma}(n,n)\lp\Omega(Y,W)\Omega(X,Z)-\Omega(X,W)\Omega(Y,Z)\rp.
\end{align}
By virtue of the definition of $A_{\widehat{h}}$, it holds
\begin{equation}
\nn (\whD_YA_{\widehat{h}})(W,X,Z)=(\whD_Y\widehat{h})(W,A(X,Z))+\widehat{h}(W,\whD_YA(X,Z)).
\end{equation}
This allows us to rewrite \eqref{mideq21} as \eqref{mideq22}.
\end{proof}
In abstract index notation the generalized Gauss identity \eqref{mideq22} takes the form
\begin{align}
\nn v_A^a\widehat{\gamma}_{af}(R^{\whnabla}){}^f{}_{bcd}v_B^bv_C^cv_D^d=&\spc \widehat{h}_{FA}(R^{\whD}){}^F{}_{BCD}+\whD_CA_{\widehat{h}}{}_{ABD}-\whD_{D}A_{\widehat{h}}{}_{ABC}\\
\nn & +\widehat{h}_{FL}A^L{}_{AD}A^{F}{}_{BC}-\widehat{h}_{FL}A^L{}_{AC}A^{F}{}_{BD}+v_D^d(\whnabla_d\widehat{\gamma}_{af})v_A^a\mathcal{P}^{f}{}_{BC}\\
\label{gammaRiemanvvvGENERAL} &-v_C^c(\whnabla_c\widehat{\gamma}_{af})v_A^a\mathcal{P}^{f}{}_{BD}+\widehat{\gamma}(n,n)\lp \Omega_{AD}\Omega_{BC}-\Omega_{AC}\Omega_{BD}\rp,
\end{align}
where the vectors $v_A^a$ are the push-forward with $\psi$ of any basis vectors  $\{\hat{v}_A\}$ in $S$.

\section{Curvature of the metric hypersurface connection $\nablao$: null case}\label{app:curvature:null}

In this appendix we write down a number of identities involving the curvature and Ricci tensors $\Riemo{}^d{}_{abc}$ and $\Riemo_{ac}\defi\Riemo{}^b{}_{abc}$ of
the metric hypersurface connection $\nablao$ (see Section \ref{sec:FHD:Prelim}). %a null metric hypersurface set $\metdata$ (see Section \ref{subsec:MHD}).

The contractions $\Riemo{}^d{}_{acb}n^a$ and
$\ell_d \Riemoin{}^d{}_{bca} n^c$ were already
computed for general hypersurface data in \cite[Eq.\ (5.3)-(5.4)]{mars2020hypersurface}. Their particularization to the null case reads
\begin{align}
    \Riemo{}^d{}_{acb} n^a &= 2 n^d \left ( \nablao_{[c} \sone_{b]} +
    P^{af} \U_{a[c} F_{b]f} \right )
    + 2 P^{df} \left ( \nablao_{[c} \U_{b]f} 
    - \sone_{[c} \U_{b]f} \right ),
    \label{Riemon_first}\\
\label{Riemoellnnull}
\ell_d \Riemoin{}^d{}_{bca} n^c & =
  \nablao_b \sone_a -\sone_b\sone_a
  + n( \elltwo)\U_{ba}+\elltwo(\pounds_n\bU)_{ba} 
  +( \F_{af} -\elltwo \U_{af}) P^{cf}\U_{bc}.
\end{align}
In the body of the paper we also need the contraction
$\gamma_{fd}\Riemo{}^d{}_{bca}n^c$.\ To compute this it is convenient to use the tensor ``Lie derivative of $\nablao$ along $n$" (we refer to \cite{yano1957Lie} for details of this tensor).\ More concretely,  let us define the
$(2,1)$-tensor 
$${\sigo} (X,W) \defi  \pounds_{n}\nablao_XW-\nablao_X\pounds_{n}W-\nablao_{\pounds_nX}W,$$ where $X,W\in\Gamma(T\N)$.\   
%In particular, 
This tensor 
%e tensor $\sigo$ 
is related to certain components of the curvature tensor $\Riemo{}^{a}{}_{bcd}$ by the identity
$\Riemo{}^d{}_{bca}n^c={\sigo}{}^d{}_{ab}-\nablao_a\nablao_b n^d$ \cite[Eq.\ (2.23)]{yano1957Lie}.\ The tensor ${\sigo}$ is 
%automatically 
symmetric as a consequence of
$\nablao$ being torsion-free.\ Moreover, it satisfies 
\cite[Eq.\ (4.9)]{yano1957Lie}
\begin{align}
\label{Sigma:eq:proof:1}0 =&\spc \mQ_{abc} -  {\sigo}{}^{d}{} _{ab} \gamma_{d c}-  {\sigo}{}^{d}{}_{ac} \gamma_{b d},\qquad\text{where}\qquad \mQ_{abc}\defi \nablao_{a} \pounds_n \gamma_{bc}-\pounds_n \nablao_{a} \gamma_{bc}.
\end{align}
With a standard trick of index permutation, it is 
straightforward to get
\begin{align}
\gamma_{fd}{\sigo}{}^{d}{}_{ab}  &=\frac{1}{2}\lp \mQ_{abf}+\mQ_{bfa} -\mQ_{fab}\rp.
\label{gammasigo}
\end{align}
Now we have the necessary ingredients to compute the contraction  $\gamma_{fd}\Riemo{}^d{}_{bca}n^c$.
\begin{lemma}
  Let  $\metdata$ be null metric hypersurface.\ Then,
\begin{align}
\label{gammaRiemonthird2} \hspace{-0.3cm} \gamma_{fd}\Riemo{}^d{}_{bca}n^c & = 
  \nablao_b\U_{fa}-\nablao_{f}\U_{ba}+2\sone_f\U_{ba}
  - \sone_b\U_{af}
+\ell_f(\pounds_{n}\bU)_{ba} -\ell_f P^{cd}\U_{bc}\U_{ad}.
\end{align}
\end{lemma}
\begin{proof}
Let us first obtain $\mQ_{abf}$ explicitly. Recalling \eqref{threetensors}
and \eqref{nablaogamma}, it follows
\begin{align}
\mQ_{abf} =&\spc 
2 \nablao_a \U_{bf} 
 + 2 \ell_{(b} \pounds_n \U_{f)a}+ 2 (\pounds_n \ell_{(b}) \U_{f)a}= 
2 \nablao_a \U_{bf}  +  2 \ell_{(b}  \pounds_n \U_{f)a}   +  4  \sone_{(b}  
\U_{f)a},  
\label{expret} 
\end{align}
after using \eqref{soneprop} for $\pounds_n \ellc$ in the second equality.
Inserting \eqref{expret} into 
\eqref{gammasigo} yields
\begin{align}
\gamma_{fd} {\sigo}{}^d{}_{ab} = &\spc 
\nablao_a \U_{bf} + \nablao_b \U_{fa} - \nablao_f \U_{ab} + \ell_f 
\pounds_n \U_{ab} + 2 \sone_f  \U_{ab} .
\label{gamsigo}
\end{align}
It only remains to compute $\gamma_{fd}\nablao_a\nablao_b n^d$ by combining $\gamma(n,\cdot)=0$, $\bU(n,\cdot)=0$, $\ellc(n)=1$ and  \eqref{prod3}-\eqref{prod4}: 
\begin{align}
  \nn \gamma_{fd}\nablao_a\nablao_b n^d=&\spc \nablao_a(\gamma_{fd}\nablao_b n^d)-\nablao_b n^d\nablao_a\gamma_{fd}   =\nablao_a(\gamma_{fd}P^{cd}\U_{bc})+(n^d\sone_b+P^{cd}\U_{bc})(\ell_{f}\U_{ad}+\ell_{d}\U_{af}) \\
=&\spc \nablao_a\U_{bf}+\sone_b\U_{af}+\ell_{f}P^{cd}\U_{ad}\U_{bc}. \label{gamma:nablao:nablao:n:new}
\end{align}
Inserting \eqref{gamsigo}-\eqref{gamma:nablao:nablao:n:new} into the identity
$\gamma_{fd}\Riemo{}^d{}_{bca}n^c=\gamma_{fd}{\sigo}{}^d{}_{ab}-\gamma_{fd}\nablao_a\nablao_b n^d$ \cite{yano1957Lie} proves \eqref{gammaRiemonthird2}.
\end{proof}
The contractions between the Ricci tensor 
$\Riemo_{ab}$ and  $n$ are also needed in the main body of the paper.\ We conclude this appendix with their derivation.
\begin{lemma}\label{lem:Riemo:contracted:with:n}
Let $\metdata$ be null metric hypersurface data.\ Then,  %following identities hold:
\begin{align}
  \label{Riemo(n,-)}  \Riemo_{ab}n^a&= 
  \pounds_{n}\sone_b-2P^{af}\U_{ab}\sone_{f}+P^{cf}\nablao_c\U_{bf} -\nablao_b(\textup{tr}_P\bU) + (\textup{tr}_P\bU)\sone_b,\\
  \label{Riemosym(n,-)} \Riemo_{(ab)}n^a& =
    \dfrac{1}{2}\pounds_{n}\sone_b-2P^{af}\U_{ab}\sone_{f}+P^{cf}\nablao_c\U_{bf} -\nablao_b(\textup{tr}_P\bU)+ (\textup{tr}_P\bU)\sone_b,\\
\label{Riemo(n,n)} \Riemo_{ab}n^an^b& %=\Riemo_{bd}n^bn^d=n^dP^{cf}\nablao_c\U_{df} -n(\textup{tr}_P\bU)
=-P^{ab}P^{cd}\U_{ac}\U_{bd} -n(\textup{tr}_P\bU).
\end{align}
%$n^dP^{cf}\nablao_c\U_{df}=-P^{cf}P^{db}\U_{df}\U_{cb}$
\end{lemma}
\begin{proof}
  To prove
  \eqref{Riemo(n,-)} we contract indices $d$ and $c$ in
  \eqref{Riemon_first} and use identity \eqref{contrNantisym} for $\bs{\theta}=\bsone$ together with $\bm{\sone}(n) =0$  so that
  $2 n^c \nablao_{[c} \sone_{b]} = \pounds_n \sone_b$,
  and hence
  \begin{align*}
    \Riemo_{ab}n^a\defi  \Riemo{}^c{}_{acb}n^a = &\spc\pounds_n \sone_b - P^{af} \U_{ab} \sone_f + P^{cf} \left ( \nablao_c \U_{bf}-
    \sone_c \U_{bf} \right )
    -  P^{cf} \left (\nablao_b \U_{cf}
    - \sone_b \U_{cf} \right ).
  \end{align*}
  We arrive at \eqref{Riemo(n,-)} after ellaborating the third term
  according to  $P^{cf} \nablao_b \U_{cf} =  
    \nablao_b \left ( \textup{tr}_P\bU \right )
    - (\nablao_b P^{cf}) \U_{cf}
    =     \nablao_b \left ( \textup{tr}_P\bU \right )$, the last equality being a consequence of  \eqref{nablaoP}.

    To prove \eqref{Riemosym(n,-)}
    we first note that
    \eqref{antisymRiemo}  becomes 
    $\Riemo_{(ab)}=\Riemo_{ab}-\nablao_{\lc a\rd}\sone_{\ld b\rc}$.\  Contracting with $n^a$ and using  $2 n^c \nablao_{[c} \sone_{b]} = \pounds_n \sone_b$ again establishes
    \eqref{Riemosym(n,-)}. Finally, \eqref{Riemo(n,n)} follows
    from any of the previous after noticing that
$n^bP^{cd}\nablao_c\U_{bd}=-\U_{bd}P^{cd}\nablao_cn^b=-P^{ab}P^{cd}\U_{ac}\U_{bd}$.
\end{proof}

\section{Gauge transformation results}\label{sec:gauge-fix:res} 
In this appendix we find the transformation law of various geometric quantities in a null hypersurface data.\ This is used in the main text to identify a number of ${\mathcal G}_1$-gauge invariant quantities. 

For any gauge parameters $\{z, V\}$ we introduce
\begin{align}
\label{def:gauge:w:f}\bw \defi  \gamma (V, \cdot), \qquad \f \defi  \ellc ( V),
\end{align}
which 
%by Lemma \ref{lemBestLemmaMarc}) 
is equivalent to
\begin{align}
  V^a = \f n^a + P^{ab}  w_b \label{Vdecom}
\end{align}
because $\gamma_{ab}\big(V^a-P^{ac}w_c\big)\stackbin{\eqref{prod4}}=w_b-(\delta_b^c-n^c\ell_b)w_c\stackbin{\eqref{prod1}}=0$ (hence $V^a-P^{ac}w_c$ is proportional to $n^a$), and $\ell_{a}\big(V^a-P^{ac}w_c\big)\stackbin{\eqref{prod3}}=f+\elltwo n^cw_c\stackbin{\eqref{prod1}}=f$ (so the proportionality function is precisely $f$).\ In terms of $\{\bw,\f\}$, the gauge transformations \eqref{gaugegamma&ell2} %\eqref{gaugeell}-\eqref{gaugeell2} 
become

\vspace{-0.6cm}

%\begin{small}
\noindent
\begin{minipage}[t]{0.4\textwidth}
\begin{align}
\label{gaugeellbis}\mathcal{G}_{\lp z,V\rp}\lp\ellc\rp&=z\lp\ellc+ \bw \rp,
\end{align}
\end{minipage}%
\hfill
%\vrule
\hfill
\begin{minipage}[t]{0.6\textwidth}
\begin{align}
\label{gaugeell2bis}\mathcal{G}_{\lp z,V\rp}\big(\ell^{(2)}\big)&=z^2\big( \ell^{(2)}+ 2 \f + P(\bw,\bw) \big).
%  \label{gaugeYbis}\mathcal{G}_{\lp z,V\rp}\lp \bY\rp & =
%z\bY+ \ellc\otimes_s \tdo z+\frac{1}{2}\lieo_{zV} \gamma
%= z\bY+ \ellc\otimes_s \tdo z+ z \f \bU +\frac{1}{2}\lieo_{z\tilde V} \gamma
\end{align}
\end{minipage}
\\

\begin{lemma}
\label{Gauge_trans_UFsone}
Let $\hypdata$ be null hypersurface data.\ Consider arbitrary gauge parameters $\{z,V\}$ and define the covector $\bw$ and the function $\f$ according to \eqref{def:gauge:w:f}.\ Then, the following gauge transformations hold:
\begin{align}
\mathcal{G}_{\lp z,V\rp}\lp \bU \rp
 & =  \frac{1}{z}\bU, \label{Uprime}\\
\mathcal{G}_{\lp z,V\rp}\lp \bF \rp & = z \lp \bF + \frac{1}{2}  d \bw \rp
+ \frac{1}{2} \tdo z \wedge \lp \ellc + \bw \rp, \label{Fprime} \\
\mathcal{G}_{\lp z,V\rp}\lp \bsone \rp & = \bsone +\frac{1}{2} \pounds_n \bw
+  \frac{n(z)}{2z}  \lp  \ellc+ \bw \rp
- \frac{1}{2z} \tdo z, \label{soneprime} \\
\mathcal{G}_{\lp z,V\rp}\lp \bsYn \rp & = \bsYn + \frac{1}{2z} \tdo z
+ \frac{n(z)}{2z} \lp \ellc + \bw \rp + \frac{1}{2} \pounds_n \bw 
- \bU (V, \cdot ), \label{Ynprime} \\
\mathcal{G}_{\lp z,V\rp}\lp \Q \rp & = \frac{1}{z} \lp \Q - \frac{n(z)}{z}  \rp.
\label{Qprime}
\end{align}
\end{lemma}
\begin{proof}
We use a prime to denote a gauge-transformed quantity. Recalling \eqref{gaugegamma&ell2},   
%\eqref{gaugegamma&ell2}, \eqref{gaugeell}, 
\eqref{gaugen} we get
\begin{align*}
 \bU'& = \frac{1}{2} \pounds_{n'} \gamma = 
\frac{1}{2} \pounds_{z^{-1} n} \gamma = \frac{1}{2z}\pounds_n \gamma = 
{z}^{-1} \bU. \\
\bF' &= \frac{1}{2} \tdo \ellc' =
\frac{z}{2} \lp \tdo \ellc + \tdo \bw \rp 
+ \frac{1}{2} dz \wedge \lp \ellc + \bw \rp 
= z \lp \bF + \frac{1}{2}  d \bw \rp
+ \frac{1}{2} \tdo z \wedge \lp \ellc + \bw \rp, \\
\bsone' & = i_{n'} \bF'  = z^{-1} i_n \bF' 
= \bsone + \frac{1}{2} i_n d \bw + \frac{n(z)}{2z}  \lp \ellc
+ \bw \rp - \frac{1}{2z} \tdo z,
\end{align*}
where $i_n$ denotes interior contraction in the first index and in the 
last equality we used $\bw(n)=0$. This proves \eqref{Uprime}  and \eqref{Fprime}, while
\eqref{soneprime} follows after using Cartan's formula $
\pounds_n \bw = i_n \tdo \bw + \tdo i_n \bw  = i_n \tdo \bw$.

For the transformation of $\bsYn$ we contract the first equality in \eqref{gaugeY} with $z^{-1} n$ to get
\begin{align*}
\bsYn'  & = 
\bsYn + \frac{1}{2z} \tdo z + \frac{n(z)}{2z} \ellc
- \frac{1}{2z} \gamma  \big ( \pounds_{z V} n, \cdot \big ) 
= 
\bsYn + \frac{1}{2z} \tdo z + \frac{n(z)}{2z} \ellc
+ \frac{1}{2z} \gamma  \big ( \pounds_{n} (zV), \cdot \big )  \\
& =  \bsYn + \frac{1}{2z} \tdo z + \frac{n(z)}{2z} \ellc
+ \frac{1}{2z} \pounds_{n} \lp  \gamma ( zV, \cdot ) \rp
- \frac{1}{2z} \lp \pounds_n \gamma \rp ( zV, \cdot)   
\end{align*}
after using the antisymmetry of the Lie bracket and ``integrating by parts".\ Expression \eqref{Ynprime} follows because $\gamma(V,\cdot ) = \bw$. 
%and $\gamma(n,\cdot)=0$. 
The last transformation follows at once from \eqref{Ynprime} and the definition 
$\Q = -\bsYn(n)$ 
\end{proof}
The following corollary plays a role in the main body of the paper.
\begin{corollary}
  \label{Gr-s}
The covector $\bsone - \bsYn$ has the following simple gauge behaviour
\begin{align*}
  \mathcal{G}_{\lp z,V\rp}\lp
  \bsone - \bsYn \rp = 
\bsone - \bsYn + \bU (V, \cdot) - \frac{1}{z} \tdo z.
\end{align*}
\end{corollary}

\section*{Acknowledgements}
The authors acknowledge financial support under the project PID2021-122938NB-I00 (Spanish Ministerio de Ciencia, Innovaci\'on y Universidades and FEDER ``A way of making Europe'') and  SA096P20 (JCyL). M. Manzano also acknowledges the Ph.D. grant FPU17/03791 (Spanish Ministerio de Ciencia, Innovaci{\'o}n y Universidades).

\begingroup
\let\itshape\upshape
\bibliographystyle{acm}

\bibliography{ref}

\end{document}

%% file: paper-config-newcommands2.tex
\usepackage[absolute,overlay]{textpos}

\usepackage{changepage}

\usepackage{float}
\usepackage{amsmath,amsthm,amsfonts,amssymb,amscd}
\usepackage{mathrsfs}
\usepackage{upgreek}
\usepackage{stackrel}
\usepackage{tipa}
\usepackage{accents}
\usepackage[multiple]{footmisc}

\usepackage{multirow} %for tables
\usepackage{booktabs}

\usepackage{xfrac} %For fractions 1/2
\newtheorem{theorem}{Theorem}[section]
\newtheorem{proposition}[theorem]{Proposition}
\newtheorem{lemma}[theorem]{Lemma}
\newtheorem{definition}[theorem]{Definition}
\newtheorem{corollary}[theorem]{Corollary}
\newtheorem{remark}[theorem]{Remark}

\newtheorem{setup}[theorem]{Setup}

\newtheorem*{orient*}{Orientation Condition}
\newtheorem*{junction*}{Standard Junction Conditions}

\def\ovkil{\ov{\teta}}
\def\ovkil{\ov{\teta}}
\def\teta{\eta}%{\widetilde{\eta}}
\def\kil{\teta}

\def\defi{{\stackrel{\mbox{\tiny {\textbf{def}}}}{\,\, = \,\, }}}

\def\nablao{{\stackrel{\circ}{\nabla}}}

\def\Riemo{{\stackrel{\circ}{R}}}
\def\Riemoin{\Riemo}
\def\Riemofull{{\stackrel{\circ}{\textup{Riem}}}}

\def\sone{\hat{s}}
\def\bsone{\bs{\hat{s}}}
\def\Ricc{{\mbox{\textbf{\textup{Ric}}}}}

\def\metdata{\{\mathcal{N},\gamma,\ellc,\elltwo\}}

\newcommand{\lightfrak}[1]{\scalebox{0.85}[1.2]{$\mathfrak{#1}$}}

\newcommand{\w}{\lightfrak{w}}%{\mathfrak{w}}
\newcommand{\qone}{\mathfrak{q}}
\newcommand{\bqone}{\bs{\qone}}
\def\bcalY{\bs{\calY}}%{\bs{\calY}}
\def\calY{\mathfrak{I}}%{\Pi}
\def\Kkil{\pounds_{\kil}g}%{\K}%{\K[\kil]}

\def\normal{\sigma}%{\tcr{\sigma}}%{\mathfrak{q}}
\def\bnormal{\bs{\normal}}

\def\G{\mathfrak{G}}

\def\sec{\mathfrak{P}}

\def\thi{\mathfrak{S}}

\def\bsecp{\bm{\sec}_{\parallel}}

\def\bthip{\bm{\thi}_{\parallel}}
\def\z{\hat{z}}

\def\bw{\bm{w}}
\def\bwp{\bw_{\parallel}}
\def\f{f}
\def\bsYn{\bs{\Yn}}

\def\bomegap{\bs{\omega}_{\parallel}}
\def\bT{\textbf{T}}

\def\trP{\textup{tr}_P}
\def\n{\mathfrak{n}}

\def\B{\mathcal B}
\def\C{\mathcal C}

\def\G{\mathcal G}
\def\Fcal{\mathcal F}

\def\bY{\textup{\textbf{Y}}}

\def\Y{\textup{Y}}
\def\bU{\textup{\textbf{U}}}
\def\U{\textup{U}}

\def\bF{\textup{\textbf{F}}}
\def\F{\textup{F}}

\def\Yn{r}
\def\Q{Q}%{\varkappa_n}
\def\mQ{\mathcal{H}}

\def\sigo{\stackrel{\circ}{\Sigma}}

\newcommand{\nn}{\nonumber}

\newcommand{\hate}{\hat{e}}
\newcommand{\bm}[1]{\mbox{\boldmath $#1$}}
\newcommand\ovnabla{\ov{\nabla}}

\newcommand\N{\mathcal N}
\newcommand\M{\mathcal M}

\newcommand\elltwo{\ell^{(2)}}

\newcommand\hypdata{\{ \mathcal{N},\gamma,\ellc, \elltwo, \bY\}}
\newcommand\rig{\zeta}

\newcommand\A{\mathcal A}

\def\bmell{\bm{\ell}}
\def\ellc{\bmell}

\def\nablao{{\stackrel{\circ}{\nabla}}}

\newcommand{\Rtensor}{\mathcal{R}}

\def\defi{{\stackrel{\mbox{\tiny \textup{\textbf{def}}}}{\,\, = \,\, }}}

\def\nablao{{\stackrel{\circ}{\nabla}}}

\def\Riemo{{\stackrel{\circ}{R}}}

\def\nabh{\nabla^h}
\def\Riemoin{\Riemo}
\def\Riemofull{{\stackrel{\circ}{\textup{\textbf{Riem}}}}}
\def\Ricco{{\stackrel{\circ}{\textup{\textbf{Ric}}}}}

\def\sone{s}
\def\bsone{\bs{s}}

\def\metdata{\{\mathcal{N},\gamma,\ellc,\elltwo\}}

\def\trP{\textup{tr}_P}
\def\n{\mathfrak{n}}

\def\B{A}%{\mathcal B}
\def\C{B}%{\mathcal C}

\def\G{\mathcal G}
\def\Fcal{\mathcal F}

\def\bY{\textup{\textbf{Y}}}

\def\Y{\textup{Y}}

\def\R{\mathscr{R}}
\def\whnabla{\widehat{\nabla}}
\def\whD{\widehat{D}}
\def\bUp{\bU_{\parallel}}
\def\bYp{\bY_{\parallel}}

\def\ellp{\ell}%{\ell_{\parallel}{}}
\def\bellp{\ellc_{\parallel}}

\def\bU{\textup{\textbf{U}}}
\def\U{\textup{U}}

\def\bF{\textup{\textbf{F}}}
\def\F{\textup{F}}

\def\Yn{r}
\def\Q{\kappa_n}%{\varkappa_n}

\def\bmell{\bm{\ell}}
\def\ellc{\bmell}

\def\nablao{{\stackrel{\circ}{\nabla}}}

\newcommand{\ov}{\overline}

\newcommand{\bs}{\boldsymbol}
\newcommand{\lp}{\left(}
\newcommand{\rp}{\right)}

\newcommand{\lb}{\left\lbrace}
\newcommand{\rb}{\right\rbrace}
\newcommand{\lc}{\left[}
\newcommand{\rc}{\right]}
\newcommand{\ld}{\left.}
\newcommand{\rd}{\right.}

\newcommand{\tdo}{d}%{\tdb}
\newcommand{\lieo}{\mathsterling}

\def\metdataa{\{\gamma,\ellc,\elltwo\}}

\newcommand{\tcb}{\textcolor{blue}}

\newcommand{\nullhyp}{\widetilde{\N}}

\newcommand{\spc}{\textup{ }}

\usepackage[new]{old-arrows}

\makeatletter
\newsavebox\myboxA
\newsavebox\myboxB
\newlength\mylenA

\newcommand*\xoverline[2][0.75]{%
    \sbox{\myboxA}{$\m@th#2$}%
    \setbox\myboxB\null% Phantom box
    \ht\myboxB=\ht\myboxA%
    \dp\myboxB=\dp\myboxA%
    \wd\myboxB=#1\wd\myboxA% Scale phantom
    \sbox\myboxB{$\m@th\overline{\copy\myboxB}$}%  Overlined phantom
    \setlength\mylenA{\the\wd\myboxA}%   calc width diff
    \addtolength\mylenA{-\the\wd\myboxB}%
    \ifdim\wd\myboxB<\wd\myboxA%
       \rlap{\hskip 0.5\mylenA\usebox\myboxB}{\usebox\myboxA}%
    \else
        \hskip -0.5\mylenA\rlap{\usebox\myboxA}{\hskip 0.5\mylenA\usebox\myboxB}%
    \fi}
\makeatother

 \newcounter{mnotecount}

 \newcommand{\mnote}[1]%{}
 {\protect{\stepcounter{mnotecount}}$^{\mbox{\tiny
 $\,\bullet$\themnotecount}}$ \marginpar{%\color{red}%
 \raggedright\tiny\em
 $\,\bullet$\themnotecount: #1} }

\allowdisplaybreaks